\documentclass[a4paper, 12pt]{article}

\usepackage[sort&compress]{natbib}
\bibpunct{(}{)}{;}{a}{}{,} 

\usepackage{amsthm, amsmath, amssymb, mathrsfs, multirow, url, subfigure}
\usepackage{graphicx} 
\usepackage{ifthen} 
\usepackage{amsfonts}
\usepackage[usenames]{color}
\usepackage{fullpage}

\theoremstyle{plain} 
\newtheorem{thm}{Theorem}

\newtheorem{lem}{Lemma}

\theoremstyle{definition}
\newtheorem{defn}{Definition}

\theoremstyle{remark}

\newcommand{\prob}{\mathsf{P}}

\newcommand{\unif}{{\sf Unif}}
\newcommand{\nm}{{\sf N}}

\newcommand{\chisq}{{\sf ChiSq}}

\newcommand{\RR}{\mathbb{R}}

\newcommand{\Z}{\mathscr{Z}}

\newcommand{\TT}{\mathbb{T}}
\newcommand{\PP}{\mathbb{P}}

\newcommand{\GG}{\mathbb{G}}

\newcommand{\model}{\mathscr{P}}

\newcommand{\eps}{\varepsilon}

\title{Direct Gibbs posterior inference on risk minimizers: construction, concentration, and calibration}
\author{Ryan Martin\footnote{Department of Statistics, North Carolina State University, {\tt rgmarti3@ncsu.edu}} \; and \; Nicholas Syring\footnote{Department of Statistics, Iowa State University, {\tt nsyring@iastate.edu}}}
\date{\today}

\begin{document}

\maketitle 

\begin{abstract}    
Real-world problems, often couched as machine learning applications, involve quantities of interest that have real-world meaning, independent of any statistical model.  To avoid potential model misspecification bias or over-complicating the problem formulation, a direct, model-free approach is desired.  The traditional Bayesian framework relies on a model for the data-generating process so, apparently, the desired direct, model-free, posterior-probabilistic inference is out of reach.  Fortunately, likelihood functions are not the only means of linking data and quantities of interest.  Loss functions provide an alternative link, where the quantity of interest is defined, or at least could be defined, as a minimizer of the corresponding risk, or expected loss.  In this case, one can obtain what is commonly referred to as a {\em Gibbs posterior distribution} by using the empirical risk function directly.  This manuscript explores the Gibbs posterior construction, its asymptotic concentration properties, and the frequentist calibration of its credible regions.  By being free from the constraints of model specification, Gibbs posteriors create new opportunities for probabilistic inference in modern statistical learning problems.  

\smallskip

\end{abstract}

\section{Introduction}
\label{S:intro}

A hallmark of the Bayesian framework is that it is normative.  That is, when presented with a new problem, a Bayesian can immediately carry out his analysis by, first, introducing a statistical model for the data, which entails a likelihood function and, second, introducing a prior distribution for the unknown parameters in that model.  Given these two inputs, the Bayesian can apply the familiar conditional probability formula to get a posterior distribution---proportional to the likelihood times prior---for the unknown, given the observed data, from which he can draw his inferences.  

Of the two required inputs, the prior attracts the most criticism from non-Bayesians.  But the requirement that a Bayesian must specify a likelihood can also be a serious obstacle in applications.  A major concern is that the statistical model may not be correctly specified.  In such a case, the model parameters have no real-world interpretation and, therefore, any inferences about them would be downright meaningless.  Marginal inference on certain features of the model parameters can still be carried out, but there is no reason to expect this to be reliable---in fact, inferences can be arbitrarily poor, as we discuss below.  
To address this concern, there have been substantial efforts to develop the subject of {\em Bayesian nonparametrics} \citep[e.g.,][]{hhmw2010, ghoshramamoorthi, ghosal.vaart.book}, which treats a key feature of the data-generating process, such as its density function, as the unknown about which inferences are to be drawn.  From a posterior for, say, the density function, marginal inference about any other relevant feature is straightforward, at least conceptually.  No doubt the nonparametric formulation makes the Bayesian's model more flexible and, consequently, his inferences more robust. Whether the robustness gained by going fully nonparametric is worth the added complexity is a question that deserves consideration, but that would have to be addressed on a case-by-case basis. In any case, the parametric and nonparametric formulations have an important point in common, namely, that the likelihood must fully specify everything about the data-generating process; that is, for given values of the unknowns, whether they be finite- or infinite-dimensional, new data could, at least in principle, be simulated according to the posited model.  This aspect of the Bayesian framework is restrictive when, like in those cases presented below, the unknown quantity to be inferred exists independent of or does not fully determine a statistical model.  For example, imagine the quantity of interest is a (conditional) quantile or, more generally, a minimizer of an expected loss function.  In such cases, the Bayesian framework offers no direct path to make posterior inference: only an indirect path through a model/likelihood specification and marginalization is possible.  What a non-Bayesian approach lacks in normativity compared the Bayesian approach it makes up for in its ability to directly infer relevant features of---rather than everything about---the data-generating process. 
\begin{quote}
{\em The idea that statistical problems do not have to be solved as one coherent whole is anathema to Bayesians but is liberating for frequentists.} \citep{wasserman.quote}
\end{quote}
The aim of this manuscript is to present a framework that we believe helps to liberate Bayesians from the need to specify a statistical model, creating an opportunity for direct, posterior-probabilistic inference in statistical learning problems.

Speaking of machine learning, it is often the goal in applications to estimate and make inference on a quantity of interest that is defined as a minimizer of an objective function which, itself, is defined as the expected value of a suitable loss function.  It is at least conceptually straightforward to get an empirical version of the risk function by averaging the loss function over the observed data points.  Then an estimator is readily obtained by minimizing this empirical risk function.  An upside to this approach is that it requires no model specification and hence has no risk of model misspecification bias.  A downside, however, is that this approach is largely focused on point estimation---it is not immediately clear how to quantify uncertainty for the purpose of inference, except perhaps for asymptotically approximate confidence regions.  For a direct, probabilistic quantification of uncertainty about the quantity of interest, without the introduction of a statistical model and the risk of misspecification bias, we recommend the construction of a so-called {\em Gibbs posterior distribution}.  The two primary ingredients that go into the construction of a Gibbs posterior are the empirical risk function---a combination of the data and the loss function that defines the problem---and a prior distribution about the risk minimizer; there is a third ingredient to be discussed below.  These go together in very much the same way that Bayes, back in 1763, originally combined a likelihood and prior, but here it is not based on a joint probability model.  The precise definition is given in Section~\ref{S:gibbs}.  Then probabilistic inference about the quantity of interest based on the Gibbs posterior proceeds exactly as it would based on a Bayesian posterior.  

That the Gibbs posterior distribution assigns probabilities to hypotheses about the quantity of interest does not, on its own, justify its use.  So, what makes inference based on the Gibbs posterior meaningful?  In addition to some basic principles justifying the specific definition in \eqref{eq:gibbs} below, and there are asymptotic results of varying strength and precision (Section~\ref{S:theory}) that suggest the Gibbs posterior will, with a sufficiently informative sample, concentrate its mass around the true risk minimizer.  Intuitively, this latter point implies inferences based on a Gibbs posterior cannot be misleading, e.g., point estimators derived from it cannot be far from whatever feature of the true risk minimizer they are supposed to be estimating.  

Of course, we want more from our framework of inference than ``not being misleading'' and, for this, special care is needed.  Towards this, the third ingredient in the Gibbs posterior distribution construction, left out of the explanation above, is a so-called {\em learning rate}.  Roughly speaking, this learning rate is a tuning parameter that controls the spread of the Gibbs posterior.  This does not affect the asymptotic concentration of mass claim above, but it does affect the Gibbs posterior's limiting form and, therefore, it also affects the reliability of inferences, e.g., the coverage probability of Gibbs posterior credible regions.  So this learning rate cannot be ignored, but must be treated carefully.  This is discussed in detail in Section~\ref{S:learning.rate}.  There we describe a particular algorithm designed to tune the learning rate, in a data-driven way, such that the Gibbs posterior credible regions attain the nominal frequentist coverage probability, hence providing reliable---instead of just ``not misleading''---inferences, even in finite samples. 

In Section~\ref{S:examples}, we present three numerical illustrations in common statistical or machine learning applications: quantile regression, classification, and (nonparametric) regression.  The focus of these examples is the role played by the learning rate and, more specifically, how the seemingly inconvenient need to specify the learning rate can be leveraged to obtain valid Gibbs posterior credible regions.  

Of course, there has been a surge of interest in generalized Bayes in recent years, so there is more to discuss than could be fit into this one manuscript.  In Section~\ref{S:more} we take the opportunity to mention a few of these developments that are outside the scope of the present paper, and also to list a few open problems that we believe, if solved, would make for nice contributions to the expanding literature in this direction.  

Some concluding remarks are made in Section~\ref{S:discuss} but, for us, the key take-away message is as follows.  Real-world problems often involve quantities of interest that have real-world meaning independent of a statistical model.  To avoid either risking model misspecification bias or overly complicating the model formulation, a direct, model-free attack on the quantity of interest is needed.  In statistical learning applications, often the quantity of interest is, or least can be, expressed as a minimizer of a suitable expected loss function.  This loss- rather than likelihood-focused link between the data and quantity of interest creates an opportunity for posterior-probabilistic inference, different from Bayes.  In our view, the results presented here make for a substantial first step towards ``liberating'' the Bayesian paradigm from its reliance on models for the data-generating process.

\section{Gibbs posterior distributions}
\label{S:gibbs}

\subsection{Problem setup}

Suppose we have data $T_1,\ldots,T_n$ assumed to be independent and identically distributed (iid) from some distribution $P$ supported on a set $\TT$.  Note that the individual $T_i$'s can be very general, so, for example, this setup can accommodate the typical supervised learning problem where $T_i = (X_i, Y_i)$ consists of a set of features/examples $X_i \in \RR^r$, possibly high-dimensional, and a response/label $Y_i \in \RR$.  Of course, dependence within $T_i$ is allowed, and captured by $P$, but independence between different $T_i$'s is assumed. 

The key difference between the Gibbs and Bayesian formulation is that, while the latter is likelihood-based, i.e., defined through specification of a statistical model, the former is loss function-based.  That is, define a (real-valued) {\em loss} function $(t,\theta) \mapsto \ell_\theta(t)$ on $\TT \times \Theta$ that measures the compatibility of the value $\theta$ of the quantity of interest with the data point $t$.  Common examples of loss functions include 
\[ \ell_\theta(t) = \{y - \theta(x)\}^2 \quad \text{and} \quad \ell_\theta(x) = 1\{y \neq \theta(x)\}, \quad t=(x,y), \]
where, in the former case, $\theta$ determines a regression function and its compatibility is measured by a squared-error loss and, in the latter case, $\theta$ determines a binary classifier and compatibility is measured by a 0--1 loss; here and throughout, $1(\cdot)$ denotes the indicator function.  Of course, other kinds of loss functions are possible, and we will see several such examples in what follows.  As one would expect, we want the ``loss'' to be small in a certain sense, so our goal is to minimize an expected loss, or {\em risk}, 
\[ R(\theta) = P \ell_\theta, \quad \theta \in \Theta, \]
where we use the operator notation for expected value of the random variable $\ell_\theta(T)$ with respect to $T \sim P$.  So then the quantity of interest is the risk minimizer 
\begin{equation}
\label{eq:rm}
\theta^\star \in \arg\min_{\theta \in \Theta} R(\theta),
\end{equation}
where ``$\in$'' allows for the possibility that the risk minimizer is not unique.  Since we do not know $P$, the risk function is inaccessible and, therefore, so too is the risk minimizer.  Then the goal is to make inference on the unknown $\theta^\star$ based on iid observations $T^n = (T_1,\ldots,T_n)$ from the unspecified distribution $P$.  

Towards this goal, we can proceed by first replacing the inaccessible risk function $R$ with an {\em empirical risk} 
\[ R_n(\theta) = \frac1n \sum_{i=1}^n \ell_\theta(T_i), \quad \theta \in \Theta, \]
and then estimating the risk minimizer, $\theta^\star$, by the {\em empirical risk minimizer}
\begin{equation}
\label{eq:erm}
\hat\theta_n \in \arg\min_{\theta \in \Theta} R_n(\theta). 
\end{equation}
In some contexts, the empirical risk minimizer is referred to as an {\em M-estimator}, and its statistical properties have been extensively studied; see, e.g., \citet{huber1981}, \citet{vaartwellner1996}, \citet{vaart1998}, \citet{kosorok.book}, and \citet{boos.stefanski.2013}. Beyond simply estimating the risk minimizer, our goal is to incorporate prior information, if available, and to quantify uncertainty about $\theta^\star$ with a data-dependent probability distribution on $\Theta$, called a Gibbs posterior, defined next.

The above discussion focused on situations where the quantity of interest has a concrete, real-world interpretation, e.g., $\theta$ is {\em defined} as the value that makes the misclassification error probability as small as possible. However, there are other situations in which a Gibbs posterior-based approach may have advantages.  Suppose, instead, that $\theta$ is defined as, say, a feature of the parameter of a posited statistical model.  It is not so uncommon these days for such models to have many nuisance parameters, intractable likelihood function, or some other complicating aspect.  In such cases, it may not be unreasonable to abandon the statistical model altogether and seek a direct construction of a Gibbs posterior for $\theta$. This would require ``reverse engineering'' a loss function such that $\theta$ can be re-expressed as a risk minimizer.  Examples of this reverse engineering can be found in \citet{syring.martin.image} and \citet{wang.martin.auc, wang.martin.levy}. 

Finally, there is another broad---and familiar---class of problems in which the same risk-minimization terminology and methodology can be applied.  Suppose we posit a statistical model $\model = \{P_\theta: \theta \in \Theta\}$, which could be finite- or infinite-dimensional, and assume that data $T_1,\ldots,T_n$ are iid $P_\theta$.  Then the goal is to estimate the unknown value of the posited model parameter.  However, as the too-often used quote reads, ``All models are wrong...,'' it is necessary to investigate the properties of an estimator of $\theta$ when the posited model happens to be wrong.  From this perspective, we formulate this as a risk-minimization problem where the loss function is 
\begin{equation}
\label{eq:log.loss}
\ell_\theta(t) = -\log p_\theta(t), \quad (t,\theta) \in \TT \times \Theta, 
\end{equation}
where $p_\theta$ is the density or mass function associated with $P_\theta$.  If $P \not\in \model$ denotes the true distribution, then the risk function is $R(\theta) = K(P,P_\theta)$, the Kullback--Leibler divergence of $P_\theta$ from $P$, so the inferential target, $\theta^\star$, is the parameter value that corresponds to the element in $\model$ closest to $P$ in the Kullback--Leibler sense.  Moreover, the empirical risk minimizer is the maximum likelihood estimator and, under certain conditions, consistency and asymptotic normality hold. However, valid inference on $\theta^\star$ will require some adjustments to account for the model misspecfication. 

If this misspecified posterior is equipped with a learning rate $\eta$, i.e., a power $\eta < 1$ on the likelihood function in the Bayes formulation, then the resulting {\em fractional Bayes posterior} \citep[e.g.,][]{bhat.pati.yang.fractional} coincides with a Gibbs posterior based on the log-loss \eqref{eq:log.loss}.  We do not pursue this specific instantiation of Gibbs posteriors any further in this paper, largely because we find the most compelling case for the Gibbs posterior comes from problems where there is no statistical model that directly connects the data and the quantity of interest.

\subsection{Definition}

With the quantity of interest defined via a loss function, instead of a likelihood, the construction of a genuine Bayes posterior distribution is out of reach.  However, we could easily just mimic the Bayes's formula by substituting the empirical risk in place of the negative log-likelihood.  This is precisely the Gibbs posterior, i.e., 
\begin{equation}
\label{eq:gibbs}
\Pi_n^{(\eta)}(d\theta) \propto e^{-\eta n R_n(\theta)} \, \Pi(d\theta), \quad \theta \in \Theta, 
\end{equation}
where $\Pi$ is a prior distribution and $\eta > 0$ is a so-called {\em learning rate} parameter that will be discussed in more detail below. The normalizing constant is determined by integrating the right-hand side of \eqref{eq:gibbs} with respect to $\theta$, so we are implicitly assuming integrability here; this holds, e.g., whenever $\ell_\theta$ is bounded away from $-\infty$.  Clearly, like Bayes's rule, this construction balances the prior and data contributions, with the data component dominating, at least for large $n$, thanks to the summation in the exponent.   

Our interpretation of the Gibbs posterior distribution, $\Pi_n^{(\eta)}$, defined in \eqref{eq:gibbs}, is as a measure that provides uncertainty quantification about $\theta$.  This is not so much different from the (non-subjective) interpretation of an ordinary Bayesian posterior distribution.  While the Gibbs posterior does have an interpretation as a coherent update of prior information in light of observed data---see \citet{bissiri.holmes.walker.2016} and below---we do not find this alone to be compelling justification for the use of a Gibbs posterior.  Our view is that the value of the Gibbs posterior, or any other statistical inference procedures for that matter, is determined by its operating characteristics, its frequentist sampling distribution properties. That is, the Gibbs posterior is not meaningful or useful based on its definition alone.  Instead, the Gibbs posterior is useful only when it can be established that inferences drawn based on it are reliable in a frequentist sense. This distinction is important for various reasons.  One in particular is that it affects the way we approach the learning rate selection problem; see Section~\ref{S:learning.rate} below. 

The loss and empirical risk component of the Gibbs posterior is what distinguishes it from a Bayesian posterior, so naturally that has been (and will be) our focus.  However, our alternative perspective also has unexpected implications on the prior distribution.  In a typical Bayesian setting, the quantities of interest are model parameters and they have no meaning outside that model, no real-world interpretation.  So it should come as no surprise that genuine prior information would typically be lacking for unknowns that have no real-world interpretation.  For example, suppose the posited model is a gamma distribution: where would genuine prior information about the shape parameter come from?  In the situations we have in mind, however, where the quantity of interest has a real-world interpretation, through the risk-minimizer characterization, it is not unreasonable to expect that genuine prior information might be available.  For example, in the illustration in Section~\ref{SSS:mcid} below, the parameter is closely related to the efficacy of a medical treatment, and it is reasonable to expect medical professionals have some prior information about its value for existing treatments. If this real-life information can be encoded in a prior distribution $\Pi$ and incorporated into the Gibbs posterior in \eqref{eq:gibbs}, the the Gibbs formulation has an upper hand in terms of efficiency compared to a more traditional model-based approach that would have the difficult task of re-expressing that available information about $\theta$ in terms of its model parameter. 

Having briefly explained our interpretation of the Gibbs posterior in \eqref{eq:gibbs}, we should also say a few words about another common interpretation.  In the computer science/machine learning literature, the Gibbs posterior is often understood simply as a ``randomized estimator;'' see Section~\ref{SSS:principles} below. That is, since their goal is simply to find parameter values that make the risk function small, $\Pi_n^{(\eta)}$ can be interpreted as an algorithm for generating samples $\tilde\theta$ such that $R(\tilde\theta)$, or $R_n(\tilde\theta)$, tends to be small.  An advantage of the randomized estimator over the empirical risk minimizer is that, from a computational point of view, it might be easier to simulate from the Gibbs posterior than to solve the empirical risk minimization problem.  


\subsection{FAQs}

Next, we have made a list of of some frequently asked questions (FAQs) about the Gibbs posterior properties, interpretation, and construction.

\subsubsection{Why a learning rate? Is it really needed?}

The reader might be surprised by the introduction of the learning rate $\eta$.  Why is this needed?  The quantity of interest, $\theta^\star$, is defined solely as the solution to an optimization problem, and the same $\theta^\star$ emerges as the solution if the loss function $\ell_\theta$ is replaced by $c \, \ell_\theta$ for any $c > 0$.  This scale of the loss function is irrelevant when focus is solely on the (empirical) risk minimization; but since the Gibbs posterior construction requires balancing the influence of the prior with that of the data, the scale of the loss function matters in practical applications.  Indeed, we will see below that, without a careful choice of the learning rate $\eta$, inference based on the Gibbs posterior can be unreliable.  We will discuss the learning rate selection process in detail in Section~\ref{S:learning.rate} below, but for now we need to make one important remark about this.  Note that the data do not directly carry any information relevant to $\eta$, so one {\em cannot} treat this choice in a Bayesian way, with prior-to-posterior updating.  The point is that $\eta$ is a tuning parameter, not a model parameter about which data are directly informative.  As such, we must rely on some other data-driven strategies (e.g., using resampling) to select the learning rate.

\subsubsection{Difference between Gibbs and misspecified Bayes?}

We mentioned above that one common way in which this loss function-based perspective emerges is when a model $\model = \{P_\theta: \theta \in \Theta\}$ is specified but it happens to be misspecified.  Then the Gibbs posterior in \eqref{eq:gibbs}, with $\eta=1$, is precisely the Bayes posterior for this misspecified model. The behavior of the Bayesian posterior under model misspecification has been extensively studied; see, e.g., \cite{berk1966}, \citet{bunke.milhaud.1998}, \citet{kleijn, kleijn.vaart.2012}, \citet{deblasi.walker.2013}, and \citet{rvr.sriram.martin}. Is the Gibbs posterior really any different than this?

First, even though Gibbs and misspecified Bayes share a resemblance, one key difference is the learning rate, $\eta$.  In correctly-specified model cases, the balance between the contributions from prior and data is automatic.  In misspecified model cases, however, this balance is thrown off and needs to be corrected manually through the insertion of a learning rate $\eta \neq 1$.  Ignoring this learning rate adjustment can lead to posterior credible intervals that have arbitrarily bad frequentist coverage probability; see, e.g., Example~2.1 in \citet{kleijn.vaart.2012}.  So if one cares about the reliability of their inferences, the learning rate must be considered and, therefore, a distinction between the Gibbs and misspecified Bayes perspectives is necessary.  

Second, the Gibbs framework is most ideally suited for cases in which no statistical model is assumed, so one cannot even ask if the model is misspecified or not.  Setting aside those examples, e.g., classification, where the loss function is a defining feature of the problem, there are other examples where it can be advantageous to introduce a loss function and construct a Gibbs posterior.  A simple but important example is that of inference on the quantile of a distribution.  The quote from \citet{wasserman.quote} stated in Section~\ref{S:intro} continues:
\begin{quote}
{\em To estimate a quantile, an honest Bayesian needs to put a prior on the space of all distributions and then find the marginal posterior. The frequentist need not care about the rest of the distribution and can focus on much simpler tasks.} 
\end{quote}
That is, the Bayesian can either start by specifying a statistical model, and risk introducing model misspecification bias, or go fully nonparametric and deal with those associated challenges.  Wasserman's ``frequentist'' is not the only one with a direct solution to this problem.  There is a well-known characterization of a quantile as the minimizer of suitable expected loss, which can be used to construct a Gibbs posterior, which can be used to make direct posterior inference on a quantile, with no risk of model misspecification bias and no nuisance parameters to be marginalized over.

\subsubsection{Principles behind the Gibbs posterior construction?}
\label{SSS:principles}

The definition of the Gibbs posterior distribution in \eqref{eq:gibbs} may give the reader the impression that the Gibbs posterior is defined to mimic the Bayesian posterior, with basically the log-likelihood replaced by the empirical risk, but this is not the case: the Gibbs posterior was first constructed in a principled manner and for a specific purpose.  Actually, there are (at least) two such constructions, and here we briefly review both of these derivations.  The first construction is primarily due to \citet{zhang2006a, zhang2006b} and shows the Gibbs posterior is the minimax optimal randomized estimator with respect to the expected, posterior-averaged risk.  The second construction is from \citet{bissiri.holmes.walker.2016} and shows the Gibbs posterior can be interpreted as a coherent updating of beliefs about the parameter, much like the Bayesian posterior, but in the case the information in the data related to $\theta$ is captured via a loss function rather than a likelihood.

Following the setup in Section~\ref{S:gibbs}, a reasonable strategy for learning about $\theta$ is to estimate it by minimizing the empirical risk $R_n(\theta)$ over $\theta\in\Theta$; penalty terms may be added, but that is beyond our scope here.  Alternatively, one may define a {\em randomized estimator} of $\theta$, which is a data-dependent distribution $\widehat\Pi_n$ that depends on data $T^n$ and a prior $\Pi$.  The idea is that, if $\widehat\Pi_n$ is ``good,'' then samples from $\widehat\Pi_n$ ought to be similarly good estimators of $\theta$.  (In some cases, sampling from $\widehat\Pi_n$ is computationally easier than minimizing the empirical risk, so this may be an attractive choice.) What makes for a ``good'' randomized estimator?  One reasonable criterion would be to insist that the $\widehat\Pi_n$-average risk, $\int R(\theta) \, \widehat\Pi_n(d\theta)$, is small in some sense.  This would imply that samples from $\widehat\Pi_n$ tend to be close to $\theta^\star$.  \citet{zhang2006a, zhang2006b} showed that
\begin{equation}
\label{eq:zhang}
P^n \int (-\eta^{-1} \log Pe^{-\eta \ell_\theta} ) \, \widehat\Pi_n(d\theta) \leq P^n \Bigl\{ \int R_n(\theta) \, \widehat\Pi_n(d\theta) + (\eta n)^{-1}K(\widehat\Pi_n, \Pi) \Bigr\}, 
\end{equation}
where the outer expectation, $P^n$, is with respect to the data $T^n$ on which $\widehat\Pi_n$ depends, and $K$ is the Kullback--Leibler divergence.  Furthermore, in certain applications it can be shown that
\[ P^n \int R(\theta) \, \widehat\Pi_n(d\theta) \leq P^n \int (-\eta^{-1} \log P e^{-\eta \ell_\theta}) \, \widehat\Pi_n(d\theta), \quad \text{for all small $\eta$}, \]
so that Zhang's bound also bounds the expected, posterior-averaged risk.  Zhang suggests the randomized estimator $\widehat\Pi_n$ be chosen so that the bracketed term in the upper bound in \eqref{eq:zhang} is made as small as possible---a kind of minimax optimality.  It is straightforward to show the minimizer of Zhang's bound is precisely the Gibbs posterior as defined in \eqref{eq:gibbs}.    

\citet{bissiri.holmes.walker.2016} derive the Gibbs posterior distribution as a coherent updating of beliefs from a prior distribution $\Pi$ on $\Theta$ to a posterior distribution $\widehat\Pi_n$ in light of data $T$.  In the Bayesian setting, the prior $\Pi$ and the likelihood both carry information about $\theta$, and Bayes' Rule provides the mechanism for combining these two sources of information.  When the likelihood is not available, but a loss $\ell_\theta$ measuring agreement/discrepancy between data and parameter is available, then \citet{bissiri.holmes.walker.2016} argue there must still be an optimal way to combine information in the prior and loss to provide a posterior update to beliefs about $\theta$.  This optimal update $\widehat\Pi_n$ should be the posterior distribution that best matches some combination of loss function and prior, or, in other words, minimizes a discrepancy between the posterior and loss, and the posterior and prior.  They argue that, for coherence of the resulting posterior, this discrepancy must be of the form 
\begin{equation}
\label{eq:bissiri}
\Psi \mapsto \int R_n(\theta) \, \Psi(d\theta) + (\eta n)^{-1} K(\Psi, \Pi).
\end{equation}
Minimizing this discrepancy is equivalent to minimizing Zhang's upper bound, and implies the coherent updating rule they seek is the same Gibbs posterior distribution in \eqref{eq:gibbs}.

Neither derivation above sheds light on the choice of learning rate $\eta$.  Any $\eta>0$ leads to a coherent updating of beliefs, and, likewise, the Gibbs posterior for any learning rate $\eta>0$ minimizes Zhang's upper bound of the expected annealed risk, since the latter is a function of $\eta$.  It must be that other considerations are needed to guide the choice of learning rate, and we discuss these in Section~\ref{S:learning.rate}.

\subsection{Illustrations}

\subsubsection{Quantiles}
\label{SSS:quantile}

Wasserman, in the quote above, brought up the example of inference on a quantile, so we take this as our first illustration.  Specifically, suppose we have data $T^n = (T_1,\ldots,T_n)$ iid from distribution $P$, and the quantity of interest is $\theta=\theta(P)$, the $\tau^\text{th}$ quantile of $P$.  A proper Bayesian solution requires that we put a prior distribution on $P$---either through a parametric model, $P_\zeta$, and a prior on $\zeta$, or through a nonparametric model---and then get the corresponding marginal posterior for $\theta$.  To simultaneously avoid the risk of model misspecification bias from introducing a parametric model and the computational challenges of working with a prior for the infinite-dimensional $P$, we proceed to construct a Gibbs posterior for $\theta$.  This requires a characterization of $\theta$ as the minimizer of a suitable risk (expected loss) function.  But recall that the $\tau^\text{th}$ quantile of a distribution can be expressed as the minimizer of $R(\theta) = P \ell_\theta$, where 
\begin{equation}
\label{eq:check}
\ell_\theta(t) = (t-\theta) \, (\tau - 1_{t \leq \theta}).
\end{equation}
This immediately leads to a Gibbs posterior distribution, with a density (with respect to Lebesgue measure) given by 
\begin{equation}
\label{eq:gibbs.density}
\pi_n^{(\eta)}(\theta) \propto e^{-\eta \, nR_n(\theta)} \, \pi(\theta), \quad \theta \in \Theta, 
\end{equation}
where $R_n$ is the empirical risk corresponding to the loss function in \eqref{eq:check}, and $\pi$ is a prior density for $\theta$ supported on $\Theta \subseteq \RR$.  

Since this is a scalar parameter problem and the loss function is relatively simple, it is easy to visualize the Gibbs posterior density.  Here we focus on visualizing the role played by the learning rate on the spread of the Gibbs posterior density.  Set $\tau=0.7$, so that interest is in inference on the $70^\text{th}$ percentile.  We simulate data from a gamma distribution with shape parameter 5 and scale parameter 1, so the true quantile is $\theta^\star \approx 5.89$.  Plots of the Gibbs posterior density for $\theta$ are shown in Figure~\ref{fig:quantile.illustration} for two sample sizes $n$ and three different learning rate values, $\eta \in \{0.1, 0.5, 1.0\}$.  The first point that deserves mention is that the Gibbs posterior densities are roughly centered around $\theta^\star$ and, as expected, they become more concentrated as sample size increases.  Second, as the learning rate varies, the center of the Gibbs posterior is mostly unchanged, but smaller $\eta$ clearly increases the posterior spread.  This highlights the impact of the learning rate on the quality of our Gibbs posterior inference. 

\begin{figure}[t]
\begin{center}
\subfigure[$n=25$]{\scalebox{0.6}{\includegraphics{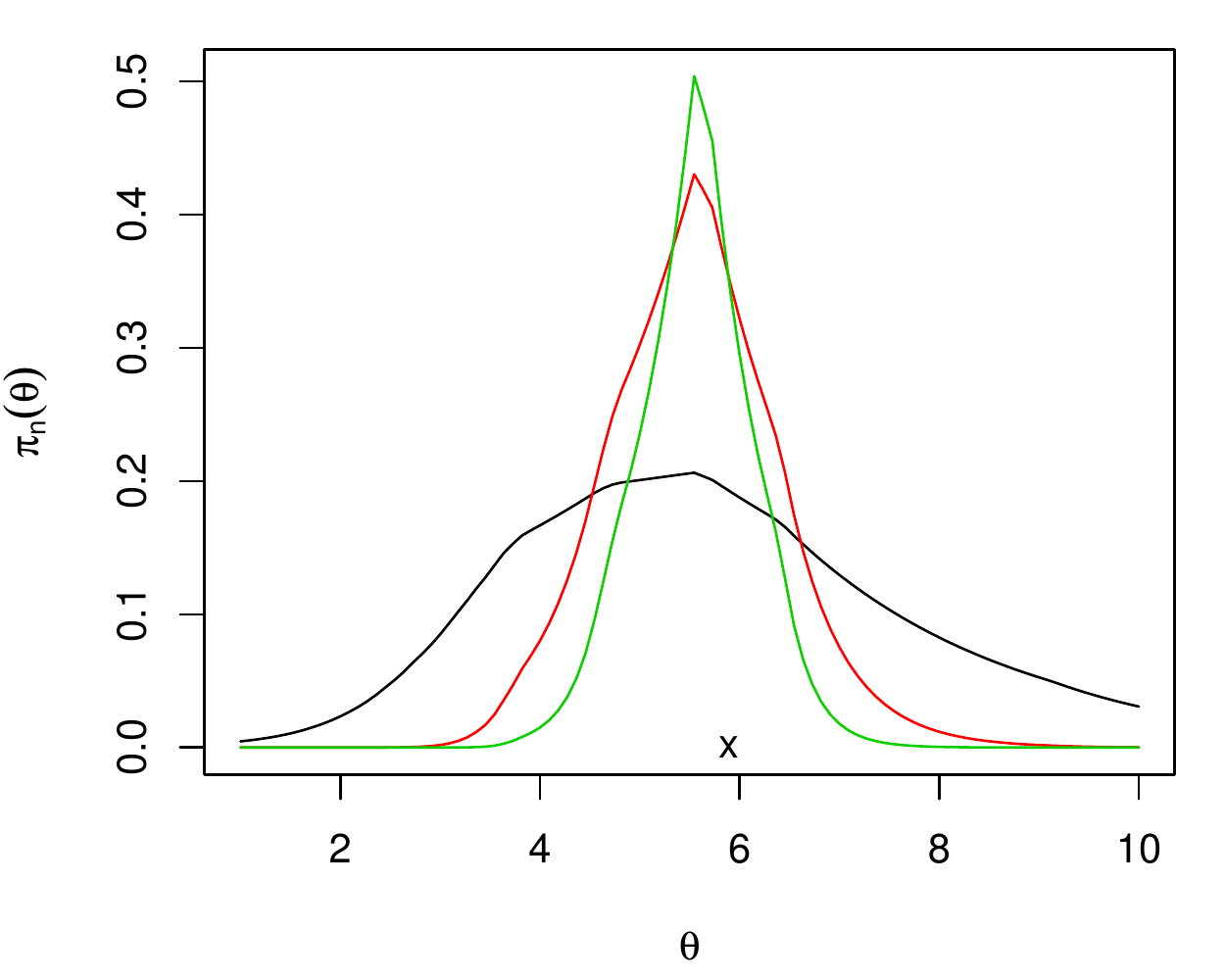}}}
\subfigure[$n=50$]{\scalebox{0.6}{\includegraphics{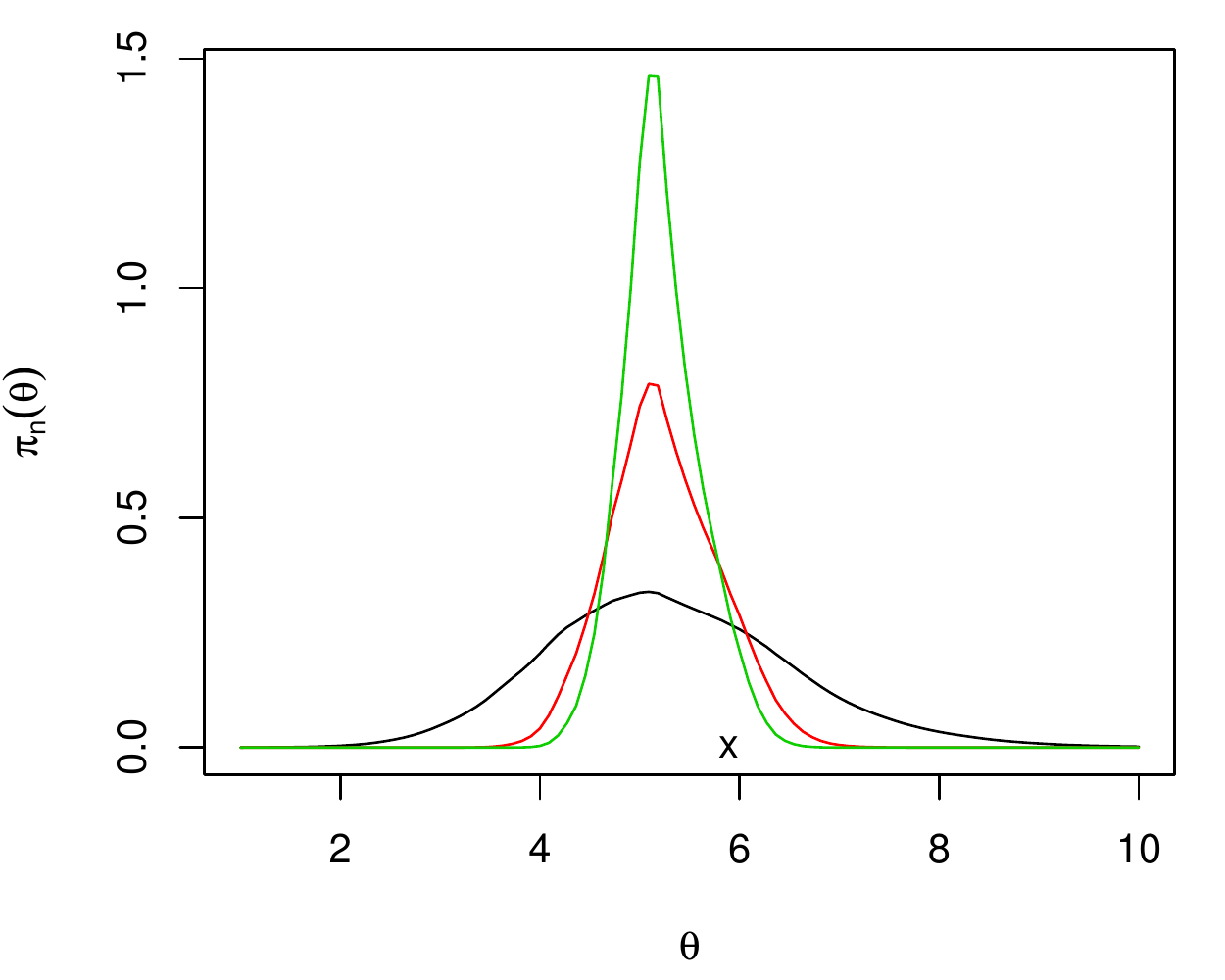}}}
\end{center}
\caption{Plots of the Gibbs posterior density for the $\tau=0.7$ quantile for two different sample sizes and three different learning rates: $\eta=0.1$ (black), $\eta=0.5$ (red), and $\eta=1.0$ (green). True $\theta^\star \approx 5.89$ is marked by ``x'' on the $\theta$-axis.}
\label{fig:quantile.illustration}
\end{figure}

\subsubsection{Minimum clinically important difference}
\label{SSS:mcid}

Next, we revisit the example that was the genesis of our investigations into Gibbs posteriors.  In a medical context, imagine that patients are given a treatment and the goal is to determine if the patients' circumstances have significantly improved from pre- to post-treatment.  A standard---and purely statistical---approach would be to determine a cutoff such that, if the patients' observed change in, say, blood pressure, exceeds that cutoff, then it is determined that this cannot be attributed to chance alone and, therefore, the treatment is judged to be statistically significant.  But it is well-known that {\em statistical significance does not imply clinical significance}, so it may be necessary to consider the latter directly.  One way to do so is to redefine the aforementioned cutoff so that it incorporates the patients' experience, and that cutoff is commonly referred to as the {\em minimum clinically important difference}, or {\em MCID}.  \citet{xu.mcid} formalized this as follows.  In addition to a real-valued, patient-specific diagnostic measure $X$ (e.g., pre- to post-treatment change in blood pressure), suppose we also have a binary, patient-reported assessment $Y$, where $Y \in \{-1,+1\}$, with $Y=+1$ if the patient felt the treatment was effective and $Y=-1$ otherwise. Then the data consists of pairs $T=(X,Y)$, having distribution $P$, and the MCID threshold is defined as 
\[ \theta^\star = \arg\min_\theta P\{Y \neq \text{sign}(X-\theta)\}. \]
That is, the MCID $\theta^\star$ is the cutoff $\theta$ on the diagnostic measure scale that makes $\text{sign}(X-\theta)$ as best a predictor of $Y$ as possible.  Equivalently, this can be expressed as the minimizer of a risk function $R(\theta) = P \ell_\theta$, as in \eqref{eq:erm}, with the corresponding loss function 
\[ \ell_\theta(t) = \tfrac12\{1 - y \, \text{sign}(x-\theta)\}, \quad t=(x,y). \]
Given iid data $T_i = (X_i,Y_i)$, for $i=1,\ldots,n$, from $P$ as described above, the MCID can be estimated by minimizing the empirical risk $R_n(\theta) = n^{-1} \sum_{i=1}^n \ell_\theta(T_i)$.  Alternatively, in \citet{syring.martin.mcid}, we constructed a Gibbs posterior distribution as in \eqref{eq:gibbs}.  

An important point we want the reader to take away from this example is that, unlike the previous example, where the empirical risk could be interpreted as a ``log-likelihood'' with respect to an asymmetric Laplace model, the Gibbs posterior solution in the MCID cannot be understood as a Bayesian solution. That is, there is no plausible model for the data having negative log-likelihood equal to $R_n(\theta)$ above.  A data analyst wishing to take a Bayesian approach to make inference on the MCID problem will likely treat $\theta$ as a feature of a statistical model.  For example, he might take a relatively simple and familiar approach through a binary regression model with, say, the logit link.  The Bayesian solution requires prior specifications for (and posterior computations of) the intercept and slope parameters, $(\alpha,\beta)$, and then marginal inference on the MCID defined as $\theta=-\alpha/\beta$.  Alternatively, to guard against potential biases due to model misspecification, the data analyst might settle on a nonparametric formulation in which the link function, say, $g(x)$, in the binary regression is itself the model parameter and a more complicated version of the prior specification and posterior computation for $g$, followed marginalization $g \to \theta$ must be carried out.  We attempted both of these strategies and found that, for inference on the MCID $\theta$, the former would often be biased while the latter would often sacrifice efficiency; see, e.g., \citet{syring.martin.mcid}, Figure~1. 

Moreover, this distinction between the Gibbs and Bayesian solutions is important in terms of how the learning rate is treated. Figure~\ref{fig:mcid.illustration}(a) shows plots of the Gibbs posterior density for the MCID for three different learning rate values.  Note the significant effect the learning rate has on how concentrated the Gibbs posterior is around the empirical risk minimizer.  Clearly the learning rate is crucial to the method's validity that the learning rate choice be handled with care.  

To help drive the latter point home, we present the results of a brief simulation study in Figure~\ref{fig:mcid.illustration}(b).  There we plot the (Monte Carlo estimate of the) coverage probability of 95\% Gibbs posterior credible intervals for $\theta$ as a function of the learning rate $\eta$.  As expected, we find that relative small (resp.~large) learning rates lead to credible intervals that over (resp.~under) cover.  But the meaning of ``relatively small/large'' changes with the sample size $n$, i.e., the learning rate value needed to hit the nominal coverage probability exactly depends on $n$---larger $n$ requires smaller $\eta$.  Since the distribution $P$ is unknown, the Monte Carlo computations done in Figure~\ref{fig:mcid.illustration}(b) cannot be carried out in practice, so a data-driven learning rate selection procedure is requires; see Section~\ref{S:learning.rate}.  

\begin{figure}[t]
\begin{center}
\subfigure[Posterior densities]{\scalebox{0.6}{\includegraphics{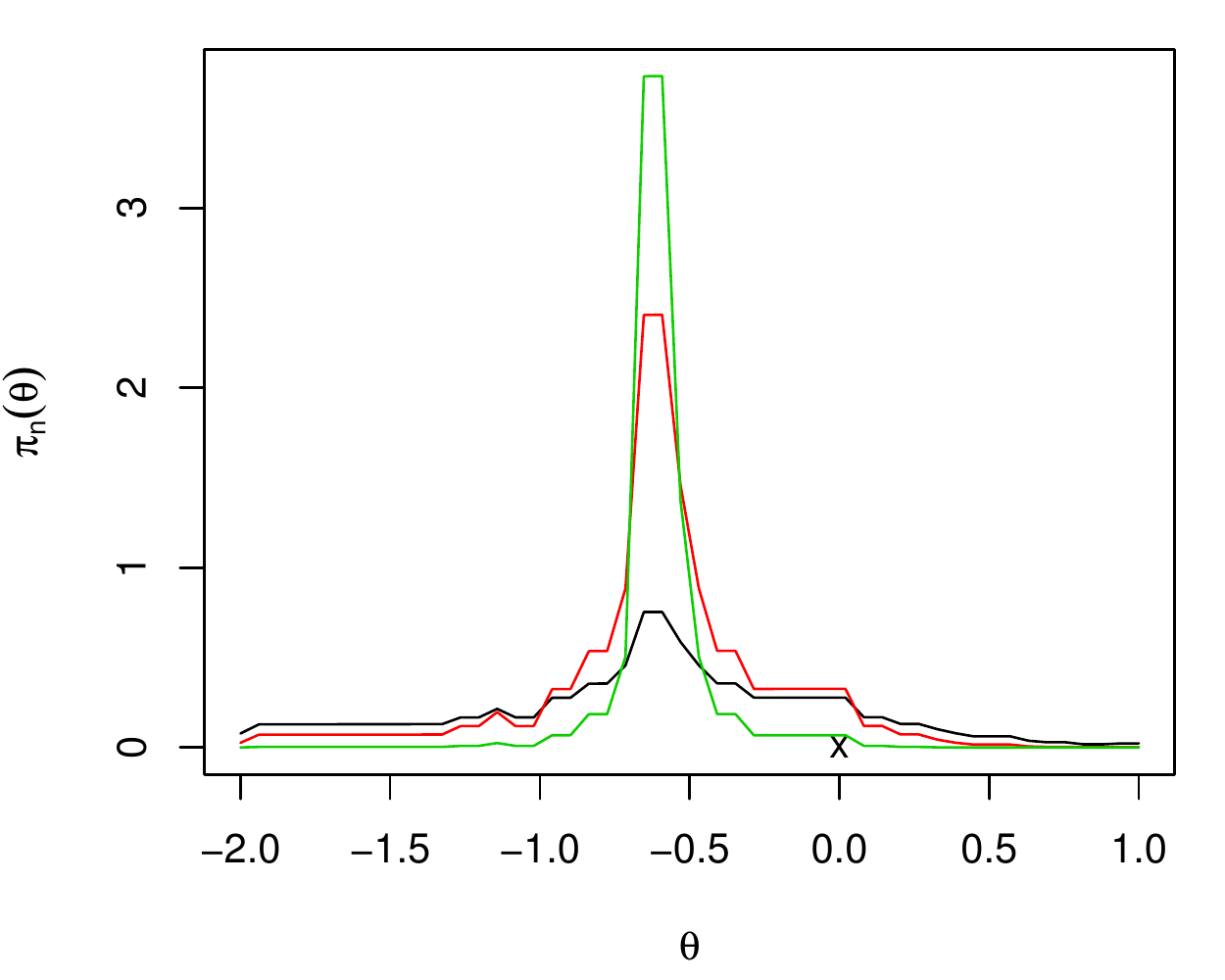}}}
\subfigure[Coverage versus $\eta$]{\scalebox{0.6}{\includegraphics{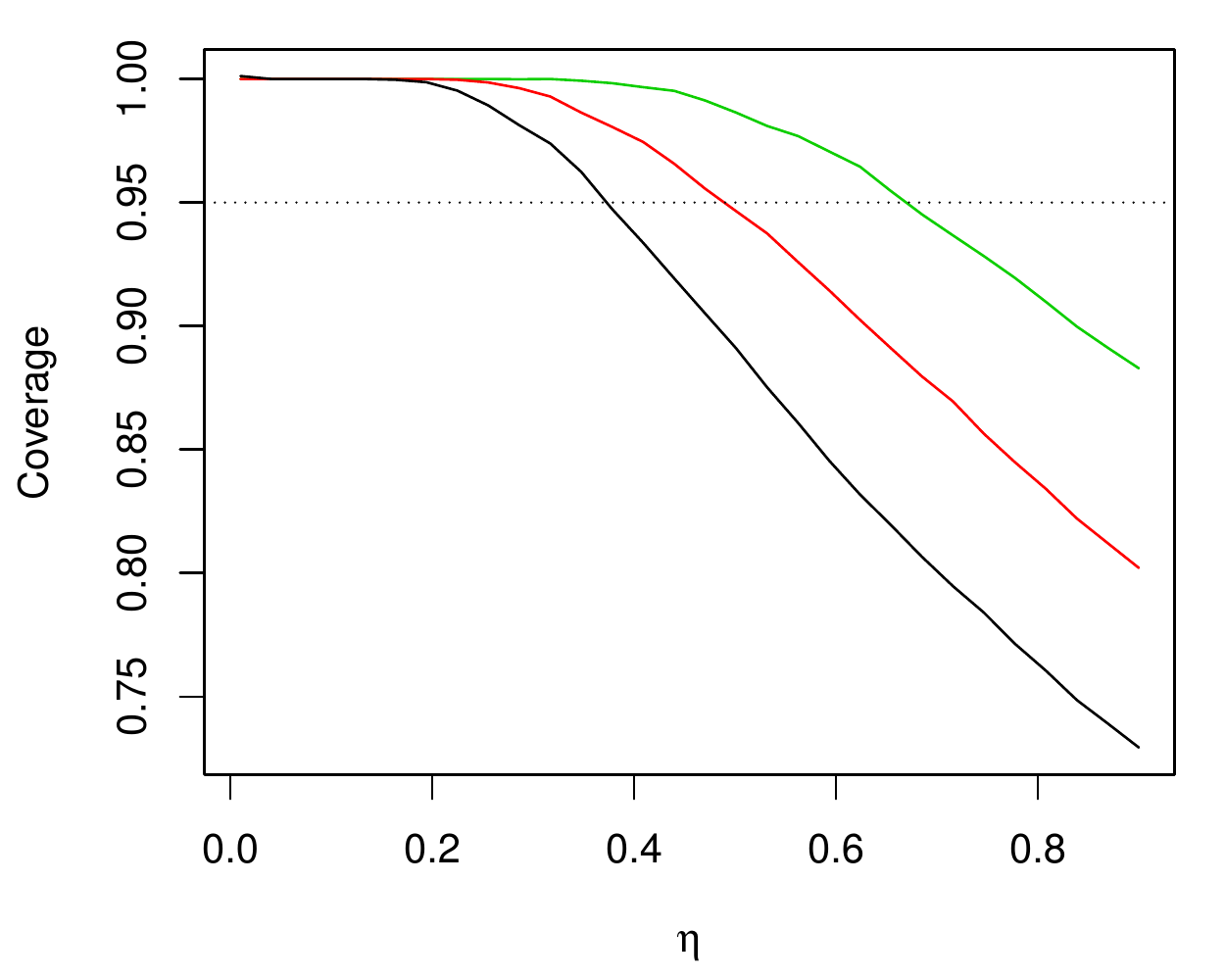}}}
\end{center}
\caption{Plots associated with the MCID example in Section~\ref{SSS:mcid}.  Panel~(a) shows the Gibbs posterior density for a simulated data set with three learning rates: $\eta=0.25$ (black), $\eta=0.5$ (red), and $\eta=1.0$ (green); true $\theta^\star \approx 0$ is marked by ``x.'' Panel~(b) shows the coverage probability of 95\% Gibbs posterior credible regions as a function of the learning rate $\eta$, for $n=50$ (green), $n=100$ (red), and $n=200$ (black).}
\label{fig:mcid.illustration}
\end{figure}

\section{Asymptotic theory}
\label{S:theory}

\subsection{Objectives and general strategies}
\label{SS:theory.basics}

If our intention is to use the Gibbs posterior, $\Pi_n^{(\eta)}$, in \eqref{eq:gibbs}, for making inference on the risk minimizer, then a first basic requirement is that it should concentrate its mass, at least as $n \to \infty$, around the true value $\theta^\star$.  More specifically, if $d$ is a suitable distance (or divergence of some kind) on $\Theta$, then we would expect the Gibbs posterior to satisfy 
\begin{equation}
\label{eq:convergence}
\Pi_n^{(\eta)}(\{\theta: d(\theta,\theta^\star) > \eps\}) \to 0 \quad \text{as $n \to \infty$ for any $\eps$}. 
\end{equation}
The reader should keep in mind that $\Pi_n^{(\eta)}$ depends on data, so it is a random measure and the convergence above is in a stochastic sense.  This will be made more precise below.  

To start, we want to provide some intuition as to why we would expect the above concentration property to hold.  Consider the case where $\Pi_n^{(\eta)}$ has a density with respect to Lebesgue measure, which we denote by $\pi_n^{(\eta)}$, as given in \eqref{eq:gibbs.density}.  If we ignore the influence of the prior distribution, which is not unreasonable since $n R_n(\cdot)$ becomes more and more influential as $n \to \infty$, then we find that the Gibbs posterior density will be maximized at $\theta=\hat\theta_n$, the empirical risk minimizer.  And thanks to normalization, it will tend to concentrate its mass around the point at which the density is maximized.  Since $R_n(\cdot) \approx R(\cdot)$ pointwise and often uniformly, we expect $\hat\theta_n \approx \theta^\star$.  Therefore, $\Pi_n^{(\eta)}$ is expected to concentrate its mass around $\theta^\star$.  The technical details that follow explain how these expectations become reality. 

It was mentioned above that our perspective on the Gibbs posterior, which we consider to be a ``statistical perspective,'' differs from that commonly taken in the computer science/machine learning literature.  Indeed, if the Gibbs posterior is viewed simply a randomized estimator of the risk minimizer, then the only thing that matters is the posterior probability assigned to risk difference neighborhoods, i.e., $\{\theta: R(\theta) - R(\theta^\star) \leq \delta\}$.  From our perspective, the Gibbs posterior can be used for general uncertainty quantification so there are other geometrically more natural metrics to consider and even more refined distributional properties, e.g., asymptotic normality, that would be both interesting and relevant.  Below we present the results from our statistical perspective, and remark on the implications for those who focus solely on risk minimization.  



\subsection{Consistency}
\label{ss:cons}

Consistency for a Gibbs posterior distribution means something very similar to consistency for a point estimator---both imply a certain random variable converges in probability to the ``right" value.  This is made precise in

\begin{defn}
\label{defn:consistency}
For a given divergence $d:\Theta\times\Theta \to \mathbb{R}^+$, the Gibbs posterior distribution $\Pi_n^{(\eta)}$ is {\em consistent} at $\theta^\star$ if 
\begin{equation}
\label{eq:consistent}
\Pi_n^{(\eta)}(\{\theta: d(\theta, \theta^\star) > \eps\}) \to 0, \quad \text{in $P$-probability, as $n \to \infty$}. 
\end{equation}
\end{defn}

Of course, since the Gibbs posterior probability is bounded, the ``in $P$-probability'' convergence in \eqref{eq:convergence} implies convergence in $P$-expectation or ``in $L_1(P)$.''  Strategies that aim to bound the $P$-expectation of the event in \eqref{eq:consistent}, rather than the $P$-probability of that event, may produce additional benefits, namely, finite-sample bounds.  We do not consider such strategies here, but see \citet{gibbs.general}.
This notation of consistency also depends on the choice of divergence $d$. Common choices include a natural/generic metric on $\Theta$, such as Euclidean distance, or a problem-specific divergence such as $d(\theta,\theta^\star) = \{ R(\theta) - R(\theta^\star)\}^{1/2}$.  Because of this dependence, we should technically write that \eqref{eq:consistent} implies $\Pi_n^{(\eta)}$ is {\em $d$-consistent}, but here $d$ will be taken as a given feature of the problem and left implicit in the notation. 


In the classical consistency results for M-estimators, or empirical risk minimizers, as presented in \citet[][Sec.~5.2]{vaart1998}, there are two key sufficient conditions.  In words, first a uniform law of large numbers is needed to ensure that the risk function to be minimized can be estimated accurately, uniformly over $\Theta$; second, a separation or identifiability condition is needed to ensure that the risk minimizer can be identified.  These two conditions, in mathematical detail, are presented in \eqref{eq:num_cons_a} and \eqref{eq:num_cons_b}, respectively:
\begin{subequations}
\label{eq:num_cons}
\begin{align}
\sup_{\theta \in \Theta} |R_n(\theta) - R(\theta)| & \to 0 \quad \text{in $P$-probability} \label{eq:num_cons_a} \\
\inf_{\theta: d(\theta,\theta^\star) > \delta} \{R(\theta) - R(\theta^\star)\} & > 0 \quad \text{for any $\delta > 0$}. \label{eq:num_cons_b}
\end{align}
\end{subequations}
Interestingly, the above sufficient conditions for M-estimator consistency turn out to be almost enough to establish Gibbs posterior consistency.  All that remains is to ensure that the prior assigns a sufficient amount of mass in the limit $\theta^\star$, i.e., 
\begin{equation}
\label{eq:prior_cond}
\Pi(\{\theta: R(\theta) - R(\theta^\star) < \delta\}) >0 \quad \text{for all $\delta>0$}.
\end{equation}
This is a very mild condition and, for example, would be satisfied in finite-dimensional settings where the prior has a strictly positive density in a neighborhood of $\theta^\star$.  

\begin{thm}
\label{thm:cons}
If \eqref{eq:num_cons} and \eqref{eq:prior_cond} hold, then the Gibbs posterior as defined in \eqref{eq:gibbs}, is consistent in the sense of Definition~\ref{defn:consistency}. 
\end{thm}

\begin{proof}
See Appendix~\ref{proof:thm1}.
\end{proof}

Since the conditions in \eqref{eq:num_cons} are identical to a common set of sufficient conditions for consistency of the M-estimator, they can be checked in all sorts of practically relevant examples.  For instance, in the quantile illustration presented in Section~\ref{SSS:quantile}, condition \eqref{eq:num_cons_b} holds if $P$ admits a unique $\tau^\text{th}$ quantile, and condition \eqref{eq:num_cons_a} holds, at least for compactly supported $P$, by the Glivenko--Cantelli theorem, a consequence of the fact that $\theta \mapsto \ell_\theta$ is Lipschitz \citep[e.g.,][Example~19.7]{vaart1998}.

\subsection{Concentration rates}
\label{ss:conc}

A more refined characterization of the asymptotic behavior of the Gibbs posterior distribution can be described by its concentration rate.  Roughly, the Gibbs posterior has concentration rate $\eps_n$ if radius-$\eps_n$ neighborhoods of $\theta^\star$ have vanishing $\Pi_n^{(\eta)}$-probability.  

\begin{defn}
\label{defn:rate}
For a vanishing sequence $\eps_n > 0$, the Gibbs posterior distribution $\Pi_n^{(\eta)}$ has concentration rate $\eps_n$ at $\theta^\star$ if 
\begin{equation}
\label{eq:rate}
\Pi_n^{(\eta)}(\{\theta: d(\theta, \theta^\star) > M_n \eps_n \}) \to 0, \quad \text{in $P$-probability, as $n \to \infty$}, 
\end{equation}
where $M_n$ is either a sufficiently large constant or a sequence diverging arbitrarily slowly.  
\end{defn} 

Comparing Definitions~\ref{defn:rate} and \ref{defn:consistency}, the former has a shrinking radius while the latter has a fixed radius, so the former result is stronger. In the typical case where $\eps_n = n^{-1/2}$, \eqref{eq:rate} can be compared to, e.g., a central limit theorem-type result where the spread of the (in this case, Gibbs posterior) distribution is shrinking at rate $n^{-1/2}$.  

Since the rate result is stronger than consistency, we can expect the sufficient conditions here to be stronger and more difficult to verify than those for consistency in Section~\ref{ss:cons}.  Fortunately, like above, connections to the M-estimator asymptotic theory are available to guide us.  Following \citet[][Sec.~5.8]{vaart1998}, we require $(P,\ell_\theta)$ to be such that, for some $(\alpha,\beta)$ with $\alpha > \beta > 0$, and for all small $\delta > 0$, 
\begin{subequations}
\label{eq:conc_num}
\begin{align}
P \sup_{\theta: d(\theta,\theta^\star)<\delta} |\mathbb{G}_n(\ell_\theta - \ell_{\theta^\star})| & \lesssim \delta^\beta \label{eq:conc_num2} \\
\inf_{\theta: d(\theta,\theta^\star)>\delta}\{R(\theta) - R(\theta^\star)\} & \gtrsim \delta^\alpha, \label{eq:conc_num1} 
\end{align}
\end{subequations}
where ``$\lesssim$'' and ``$\gtrsim$'' denote inequality up to a constant multiple, and $\mathbb{G}_n$ denotes the empirical process, i.e., $\mathbb{G}_n f := n^{1/2}(\mathbb{P}_nf - Pf)$; note that, in case the random variable in \eqref{eq:conc_num2} is not measurable, the expectation can be replaced by an upper expectation.  Intuitively, at least from the M-estimation perspective, what matters is that empirical risk $R_n$ has roughly the same behavior as the risk $R$; in that case, if $R$ has a discernible minimizer, then $R_n$ will too.  While the empirical process notation complicates matters, condition \eqref{eq:conc_num2} amounts to having some uniform control on the fluctuations of $R_n$ around $R$.  Moreover, \eqref{eq:conc_num1} is a refined version of the condition \eqref{eq:num_cons_b} that provides some quantification of how discernible the minimizer of $R$ is.  If these conditions are satisfied, then the Gibbs posterior distribution will concentrate at $\theta^\star$ at a rate determined by the pair $(\alpha,\beta)$.  In particular, smaller $\alpha$ means greater discernibility, and smaller $\beta$ means tighter control on the fluctuations, which should make the concentration rate faster.  

As above, the conditions \eqref{eq:conc_num} are almost enough to establish the Gibbs posterior concentration rate.  All that remains is to ensure the the prior assigns sufficient mass to neighborhoods of $\theta$ of the appropriate radius.  This is a bit more complicated than the analogous condition \eqref{eq:prior_cond} for consistency.  To state this condition precisely, we will need the two functions 
\[ m(\theta,\theta^\star) = R(\theta) - R(\theta^\star) \quad \text{and} \quad v(\theta,\theta^\star) = P(\ell_\theta - \ell_{\theta^\star})^2 - m^2(\theta, \theta^\star), \]
and the corresponding neighborhood 
\[ \Theta(r) = \{\theta: m(\theta, \theta^\star) \vee v(\theta, \theta^\star) \leq r\}, \quad r > 0. \]

\ifthenelse{1=1}{}{
\begin{thm}
\label{thm:conc}
Suppose that $(P,\ell_\theta)$ are such that \eqref{eq:conc_num} hold with constants $(\alpha,\beta)$ satisfying either
\begin{align*}
    &\alpha \leq 1\quad\text{and}\quad \alpha - \beta > 1/2; \text{ or}\\
    &\alpha > 1\quad\text{and}\quad \alpha - 2\beta > 0.
\end{align*}
Set $\eps_n = n^{-1/(2\alpha-2\beta)}$.  If, for constants $C,C' > 0$, the prior $\Pi$ satisfies 
\begin{equation}
\label{eq:conc_den}
\Pi\{\Theta(C\eps_n)\} \gtrsim \exp(-C' n \eps_n^\alpha),
\end{equation}
then \eqref{eq:rate} holds and the Gibbs posterior $\Pi_n^{(\eta)}$ has concentration rate $\eps_n$ at $\theta^\star$.  
\end{thm}}

\begin{thm}
\label{thm:conc}
Consider a finite-dimensional $\theta$, taking values in $\Theta\subseteq \mathbb{R}^q$ for some $q\geq 1$.  Suppose $(P,\ell_\theta)$ are such that \eqref{eq:conc_num} hold with constants $(\alpha,\beta)$ satisfying $\alpha \geq 2\beta$, and define $\eps_n=n^{-1/(2\alpha - 2\beta)}$.  If the prior $\Pi$ satisfies 
\begin{equation}
\label{eq:conc_den}
\Pi\{\Theta(\eps_n)\} \gtrsim \eps_n^q,
\end{equation}
then \eqref{eq:rate} holds and the Gibbs posterior $\Pi_n^{(\eta)}$ has concentration rate $\eps_n$ at $\theta^\star$.
\end{thm}

\begin{proof}
See Appendix~\ref{proof:thm2}.
\end{proof}

In regular finite-dimensional problems, where ``regular'' means that the empirical risk has a certain degree of smoothness, the Gibbs posterior concentration rate would be root-$n$, i.e., the conditions of Theorem~\ref{thm:conc} could be checked with $\alpha=2$ and $\beta=1$, so that $\eps_n = n^{-1/2}$.   For example, in the quantile problem from Section~\ref{SSS:quantile}, since the loss is Lipschitz, control over the random fluctuations follows from standard results, e.g., Corollary~19.35 in \citet{vaart1998}.  In particular, \eqref{eq:conc_num2} holds with $\beta=1$ where $d(\theta, \theta^\star) = |\theta-\theta^\star|$.  Moreover, if $P$ admits a density function that is positive at $\theta^\star$, the risk $R$ is approximately quadratic in a neighborhood of $\theta^\star$, so \eqref{eq:conc_num1} holds with $\alpha=2$.  Putting this together, if the prior density is bounded away from 0 in a neighborhood of $\theta^\star$, then it follows from Theorem~\ref{thm:conc} that the Gibbs posterior concentrates at a root-$n$ rate.  But rates faster and slower than root-$n$ are possible outside of these ``regular'' problems.  A good example is the MCID application in Section~\ref{SSS:mcid}: as \citet{syring.martin.mcid} show, the rate can be as fast as $n^{-1}$ and as slow as $n^{-1/3}$, depending on certain features of the underlying $P$.  

The proof of Theorem~\ref{thm:conc} can easily be adapted to handle infinite-dimensional $\theta$ by modifying Lemmas~\ref{lem:conc_den}--\ref{lem:conc_num} in Appendix~\ref{S:proofs}; see, e.g., \citet{syring.martin.image, gibbs.general}.

\subsection{Distributional approximations}
\label{SS:bvm}

Beyond consistency and rates, there are cases in which the Gibbs posterior enjoys a version of the celebrated {\em Bernstein--von Mises} theorem, i.e., that the Gibbs posterior takes on a Gaussian shape asymptotically.  This was demonstrated for a special case in \citet{gibbs.quantile} but their results are generalized below.

Common folklore is that the Bernstein--von Mises theorem guarantees the Bayesian posterior will be asymptotically calibrated in the sense that its credible regions will agree with the frequentist confidence regions, hence that Bayesian inference would be at least approximately valid, in a frequentist sense, for large $n$.  This suggest a best-of-both-worlds conclusion, i.e., that one can have both the appeal of doing formal probabilistic inference with Bayes's theorem and frequentist error rate guarantees.  What makes this ``folklore'' is that it holds only in well-specified model cases.  When the model is misspecified, a Bernstein--von Mises theorem can still be established, modulo regularity conditions, but it does not enjoy the same best-of-both-worlds interpretation as in the well-specified case.  Indeed, \citet{kleijn.vaart.2012} show that, while the misspecified Bayes posterior may still be asymptotically normal, misspecification bias creates a mismatch between the limiting posterior covariance and that of the sampling distribution of the posterior mean.  This covariance mismatch implies, e.g., that the frequentist coverage probability of the Bayesian posterior credible regions can be arbitrarily far from the nominal level.


While the Bernstein--von Mises theorem in misspecified model cases does not have the same strong implications as in well-specified cases, it is still an interesting theoretical result.  Moreover, when applied to a generalized/Gibbs posterior, the result is practically relevant because it sheds light on the learning rate's role in the limiting posterior, which in turn can be informative for the data-driven tuning discussed below.  

For the situations involving iid data under consideration here, it suffices to consider those cases where the Gibbs posterior concentrates at rate $\eps_n = n^{-1/2}$.  As explained above in Section~\ref{ss:conc}, a root-$n$ rate is common in fixed, finite-dimensional problems where loss function has a certain degree of smoothness; here we let $q$ denote the finite dimension of $\theta$ so that $\Theta \subseteq \RR^q$.  This includes our simple running example of inference on a quantile, with $q=1$.  In order to say more about the limiting shape of the Gibbs posterior, even more smoothness of the loss is required.  In the classical theory of well-specified parametric Bayes models, e.g., Theorem~4.2 of \citet{ghosh-etal-book}, the sufficient conditions for the Bernstein--von Mises theorem include twice differentiability of the log-likelihood.  More modern approaches based on local asymptotic normality, e.g.,  \citet[][Ch.~6]{lecam.yang.book} and \citet[][Ch.~7]{vaart1998}, provide some additional flexibility.  Here we follow this more modern approach and assume only that 
\begin{itemize}
\item $R$ is twice differentiable at $\theta^\star$, with $\dot R(\theta^\star) = 0$ and $V_{\theta^\star} := \ddot R(\theta^\star)$ positive definite;
\item the loss function $\theta \mapsto \ell_\theta(x)$ can be differentiated for $P$-almost all $x$.
\end{itemize}
Here, dot and double-dot correspond to first and second derivatives with respect to $\theta$, so $\dot R$ and $\ddot R$ denote the gradient vector and the Hessian matrix, respectively.  Note that this approach avoids assuming the loss is twice differentiable as would be required under the classical theory.  In our case, the local asymptotic normality condition takes the form 
\begin{equation}
\label{eq:lan}
\sup_{h \in K} \bigl| D_n(h; \theta^\star) -  h^\top V_{\theta^{\star}} \Delta_{n,\theta^{\star}}-\tfrac{1}{2}h^\top V_{\theta^{\star}}h \bigr| = o_P(1), \quad \text{all compact $K \subset \RR^q$}, 
\end{equation}
where $D_n(h; \theta^\star) = n\{R_n(\theta^{\star}+h n^{-1/2}) - R_n(\theta^{\star})\}$ is the scaled local empirical risk difference at $\theta^\star$, $V_{\theta^\star}$ is as defined above, and 
\[ \Delta_{n,\theta^\star} = n^{-1/2} \sum_{i=1}^n V_{\theta^\star}^{-1} \dot\ell_{\theta^\star}(X_i). \]
The intuition behind \eqref{eq:lan} is that the empirical risk difference is approximately quadratic, locally near $\theta^\star$, but in a sense that does not require twice differentiabilty.  And, as usual, a quadratic approximation appearing in the exponent suggests a Gaussian approximation, where the matrix appearing in the quadratic term determines the Gaussian's variance.  

\begin{thm}
\label{thm:bvm}
Suppose the Gibbs posterior $\Pi_n^{(\eta)}$, for fixed $\eta > 0$, has concentration rate $n^{-1/2}$.  In addition, if the loss function $\ell_\theta$ is such that \eqref{eq:lan} holds, then the sequence of appropriately centered and scaled Gibbs posteriors approaches a sequence of $q$-variate normal distributions in total variation, as $n \to \infty$; that is,  
\begin{equation}
\label{eq:bvm}
\sup_B\bigl| \Pi_n^{(\eta)}(\{\theta: n^{1/2}(\theta-\theta^{\star}) \in B\})-\nm_q(B \mid \eta\Delta_{n,\theta^{\star}},{(\eta V_{\theta^{\star}})}^{-1})\bigr| = o_P(1). 
\end{equation}
\end{thm}

\begin{proof}
The proof is virtually identical to that of Theorem~2.1 in \citet{kleijn.vaart.2012} so we will not reproduce the details here.  The only difference is that their likelihood ratio would be replaced by our $\exp[-\eta n \{R_n(\theta) - R_n(\theta^\star)\}]$. 
\end{proof}

It would often be the case \citep[e.g.,][Theorem~5.7]{vaart1998} that the empirical risk minimizer, $\hat\theta_n$ satisfies 
\[ n^{1/2}(\hat\theta_n - \theta^\star) = \Delta_{n,\theta^\star} + o_P(1). \]
Then it follows from Theorem~\ref{thm:bvm} above and the location shift invariance of the total variation distance, that \eqref{eq:bvm} can be re-expressed as 
\begin{equation}
\label{eq:bvm2}
\sup_B\bigl| \Pi_n^{(\eta)}(B)-\nm_q(B \mid \eta \hat\theta_n + (1-\eta)\theta^\star, {(\eta nV_{\theta^{\star}})}^{-1})\bigr| = o_P(1), \quad n \to \infty.
\end{equation}

Compared to \eqref{eq:bvm}, the form in \eqref{eq:bvm2}\footnote{In \citet{gibbs.quantile}, the effect of $\eta$ on the asymptotic Gibbs posterior mean was overlooked---they stated the posterior mean was $\hat\theta_n$ instead of $\eta\hat\theta_n + (1-\eta)\theta^\star$ as in \eqref{eq:bvm2}. This small effect went unnoticed because the learning rate suggested in the former case ends up being larger than in the latter, hence more conservative Gibbs posterior credible regions.} makes it easier to see the effect of the learning rate $\eta$.  Of course, when $\eta=1$, this looks exactly like the Gaussian approximation presented \citet[][p.~362]{kleijn.vaart.2012}.  The effect of a value $\eta < 1$ on the mean is negligible, since $\hat\theta_n \approx \theta^\star$ when $n$ is large.  For the covariance, the effect of $\eta$ can be more substantial, hopefully in a good way.  Towards this, recall \citep[e.g.,][Theorem~5.23]{vaart1998, muller2013} that the empirical risk minimizer, or M-estimator, has covariance matrix $n^{-1} \Sigma_{\theta^\star}$, where $\Sigma_{\theta^\star}$ is the sandwich covariance matrix
\begin{equation}
\label{eq:sandwich}
\Sigma_{\theta^\star} = V_{\theta^\star}^{-1} \, P(\dot\ell_{\theta^\star} \dot\ell_{\theta^\star}^\top) \, V_{\theta^\star}^{-1}. 
\end{equation}
In general, $\Sigma_{\theta^\star} \neq V_{\theta^\star}^{-1}$, which is the aforementioned covariance mismatch.  However, there are cases when a so-called {\em generalized information equality} \citep{chernozhukov.hong.2003}, which states that $\Sigma_{\theta^\star} = \gamma V_{\theta^\star}^{-1}$, for some scalar $\gamma > 0$.  In such cases, there exists $\eta$ such that the covariance matrix of the Gibbs posterior mean, $\eta^2 \Sigma_{\theta^\star}$, matches the Gibbs posterior covariance matrix, $(\eta V_{\theta^\star})^{-1}$, i.e., 
\[ \eta^2 \Sigma_{\theta^\star} = \eta^{-1} V_{\theta^\star}^{-1} \iff \eta = \gamma^{-1/3}. \]
So there is a learning rate value that corrects the covariance mismatch and leads to valid Gibbs posterior inference, at least asymptotically.  More generally, we can expect that there is some value of $\eta$ for which the above relationship holds at least approximately.  This begs the question: {\em how might that learning rate value be found?}

\section{Learning rate selection}
\label{S:learning.rate}

In Sections~\ref{S:intro}--\ref{S:gibbs}, we emphasized the importance of the learning rate, but then the learning rate was mostly irrelevant in the theoretical results presented in Section~\ref{S:theory}.  The reason for this apparent discrepancy is that the theoretical results are all ``first-order'' in the sense that they only describe features of the Gibbs posterior relevant to estimation; no ``higher-order'' claims about accuracy of inference based on the Gibbs posterior have been made.  Even the Bernstein--von Mises result, despite being distributional in nature, provides no inference guarantees in our under- or misspecified model setting the way the analogous result does in a well-specified model setting.  So, what we said in Sections~\ref{S:intro}--\ref{S:gibbs} remains true: the learning rate needs to be chosen carefully in practical applications to ensure that inferences drawn based on the Gibbs posterior are reliable.  

As mentioned in Section~\ref{S:intro}, data-driven learning rate selection methods has been an active area of research in recent years.  A number of novel ideas have been put forth, from different perspectives and with distinct objectives.  This includes the methods by \citet{grunwald2012}, \citet{grunwald.ommen.scaling}, \citet{holmes.walker.scaling}, \citet{lyddon.holmes.walker}, and \citet{syring.martin.scaling}. We will focus the presentation here on our preferred method, the {\em general posterior calibration} or {\em GPC} algorithm.  This is our preferred method not just because we developed it, but also because it has the best empirical performance---in terms of coverage probability of credible sets---across various settings and sample sizes compared to the other methods; see \citet{gpc.compare}.  In what follows, we explain what the GPC algorithm aims to do, give some heuristics for why it works, and then describe the algorithm and its implementation details.  

First, we need to justify an important but basic claim, namely, that the primary role played by the learning rate, $\eta$, in the Gibbs posterior, $\Pi_n^{(\eta)}$, defined in \eqref{eq:gibbs}, is to control the spread.  As we explained in Section~\ref{SS:theory.basics}, the prior's influence will be rather limited, at least when $n$ is large, so the Gibbs posterior density is, as in \eqref{eq:gibbs.density}, effectively proportional to $\exp\{-\eta n R_n(\theta)\}$, for $\theta \in \Theta$.  Since $R_n$ is minimized at $\hat\theta_n$, independent of $\eta$, we see that $\eta$ can only be affecting the spread of the Gibbs posterior, with small $\eta$ making the posterior wider, more diffuse, and large $\eta$ making the posterior narrower, more concentrated at $\hat\theta_n$.  The posterior consistency result in Theorem~\ref{thm:cons} formalizes this. 

So if the learning rate controls the spread of the posterior, it also must influence the coverage probability of Gibbs posterior credible regions.  That is, let $C_\alpha^{(\eta)}(T^n)$ denote a $100(1-\alpha)$\% Gibbs posterior credible region, based on data $T^n$; for example, it could be a highest posterior density (HPD) region give by  
\[ C_\alpha^{(\eta)}(T^n) = \{\theta: \pi_n^{(\eta)}(\theta) > k(\alpha; \eta)\}, \]
where $\pi_n^{(\eta)}$ is the Gibbs posterior density in \eqref{eq:gibbs.density} and $k(\alpha; \eta)$ is the cutoff chosen to ensure that the region has $\Pi_n^{(\eta)}$-probability $1-\alpha$.  However, it is worth pointing out that it is not necessary that the credible region be for the entire unknown $\theta$, it could be just for some relevant feature $\psi=\psi(\theta)$.  Now define the (frequentist) coverage probability function 
\begin{equation}
\label{eq:cvgfun}
c_\alpha(\eta) = c_\alpha(\eta; P) = P\{ C_\alpha^{(\eta)}(T^n) \ni \theta(P)\},
\end{equation}
where, here, we have made the notation explicitly reflect the dependence of the inferential target ``$\theta^\star$'' on the underlying probability $P$.  Since $\eta$ controls the spread of the Gibbs posterior, it likewise controls the size of the credible regions.  Therefore, $\eta \mapsto c_\alpha(\eta)$ is decreasing, so we ought to be able to find a learning rate value that gets the coverage probability close the advertised/nominal frequentist level.  That is, we seek
\begin{equation}
\label{eq:gpc.target}
\eta^\star = \eta^\star(\alpha; P) = \sup\{\eta > 0: c_\alpha(\eta) \geq 1-\alpha\}.
\end{equation}
Such a learning rate $\eta^\star$ would {\em calibrate} the Gibbs posterior in the sense that its $100(1-\alpha)$\% credible region would have (frequentist) coverage probability at least $1-\alpha$, i.e., the uncertainty quantification would be valid or honest.  

Before moving on to describe what the {\em oracle} learning rate $\eta^\star$ looks like and how to approximate it in a data-driven way, it is important that we dispel with some of the optimism that stems from our many experiences focusing on well-specified models.  That is, at least in general, we cannot expect to find a single $\eta^\star$ to achieve the desired calibration for every $\alpha$, $n$, $P$, relevant feature $\psi=\psi(\theta)$, etc.  When dealing with an under- or misspecified model, we have to choose which battles we want to win and, in our present case, we have chosen to focus on the particular $100(1-\alpha)$\% credible region and choosing the learning rate $\eta^\star$ that ensures its coverage attains the nominal frequentist level.  For other objectives, a different $\eta^\star$ would be required.  

For those special cases described in Section~\ref{SS:bvm}, where the Gibbs posterior is asymptotically normal, we can shed some light on what the oracle $\eta^\star$ looks like.  From Theorem~\ref{thm:bvm} and, in particular, the version of the conclusion in \eqref{eq:bvm2}, we have that the $100(1-\alpha)$\% credible region associated with the Gibbs posterior distribution $\Pi_n^{(\eta)}$ has the form 
\[ \{\vartheta: n(\vartheta - \hat\theta_n^{(\eta)})^\top (\eta  V_{\theta^\star}) (\vartheta - \hat\theta_n^{(\eta)}) \leq k_\alpha\}, \]
where $k_\alpha$ is the upper-$\alpha$ quantile the chi-square distribution with $q$ degrees of freedom, and $\hat\theta_n^{(\eta)} = \eta \hat\theta_n + (1-\eta)\theta^\star$.  Then the coverage probability of the credible region is 
\[ c_\alpha(\eta) = P\{n(\hat\theta_n - \theta^\star)^\top (\eta^3 V_{\theta^\star})(\hat\theta_n - \theta^\star) \leq k_\alpha\}, \]
where the probability is with respect to the sampling distribution of $\hat\theta_n$ under $P$.  Under the assumptions of Theorem~\ref{thm:bvm}, we also have that $\hat\theta_n \sim \nm_q(\theta^\star, n^{-1}\Sigma_{\theta^\star})$, for large $n$, with the latter covariance matrix defined in Section~\ref{SS:bvm}.  So, if it happened that $\eta^3 V_{\theta^\star} = \Sigma_{\theta^\star}^{-1}$, then the coverage probability would be approximately equal to $1-\alpha$.  More generally, the quadratic form in the above display is no smaller than 
\[ n (\hat\theta_n - \theta^\star)^\top \Sigma_{\theta^\star}^{-1} (\hat\theta_n - \theta^\star) \times \eta^3 \, \lambda_{\text{min}}(\Sigma_{\theta^\star}^{1/2} V_{\theta^\star} \Sigma_{\theta^\star}^{1/2}), \]
where $\Sigma_{\theta^\star}^{1/2}$ is a suitable square root matrix of $\Sigma_{\theta^\star}$ and $\lambda_{\text{min}}$ is the minimum eigenvalue operator.  Therefore, for large $n$, the coverage probability satisfies 
\[ c_\alpha(\eta) \geq \prob\Bigl\{ \chisq(q) \leq \frac{k_\alpha}{\eta^3 \, \lambda_{\text{min}}(\Sigma_{\theta^\star}^{1/2} V_{\theta^\star} \Sigma_{\theta^\star}^{1/2})} \Bigr\}. \]
So it is clear that the set over which the supremum in \eqref{eq:gpc.target} is taken is non-empty and, moreover, the oracle $\eta^\star$ would be roughly 
\[ \eta^\star \approx \{ \lambda_{\text{min}}(\Sigma_{\theta^\star}^{1/2} V_{\theta^\star} \Sigma_{\theta^\star}^{1/2})\}^{-1/3}. \]

The obvious question is how can we find the $\eta^\star$ in \eqref{eq:gpc.target}?  Since we do not know $P$, we clearly cannot evaluate the coverage probability function $c_\alpha(\eta)$ and, therefore, we cannot solve the equation $c_\alpha(\eta) = 1-\alpha$.  We can, however, obtain roughly unbiased estimates of the coverage probability function at any fixed $\eta$ via the bootstrap \citep{davison.hinkley.1997, efrontibshirani1993, efron1979}.  That is, if we let $\tilde T_b^n$ denote an iid sample of size $n$ from the empirical distribution, $\PP_n$, of the observed data $T^n$, for $b=1,\ldots,B$, then a bootstrap approximation of the coverage probability function is 
\[ \hat c_\alpha^\text{boot}(\eta) = \frac1B \sum_{b=1}^B 1\{ C_\alpha^{(\eta)}(\tilde T_b^n) \ni \hat\theta_n\}. \]
Alternatively, this can be viewed as a Monte Carlo approximation of the plug-in estimate 
\[ \hat c_\alpha(\eta) = c_\alpha(\eta; \PP_n) = \PP_n\{ C_\alpha^{(\eta)}(\tilde T^n) \ni \theta(\PP_n)\}. \]
Since these coverage probability functions are only estimates/approximations, we need to acknowledge the variability in how we solve the equation ``$c_\alpha(\eta) = 1-\alpha$.''  For this, we apply the stochastic approximation procedure of \citet{robbinsmonro}, which chooses an initial guess $\eta_0$ and defines an sequence of candidate solutions 
\begin{equation}
\label{eq:sa.update}
\eta_s = \eta_{s-1} + \kappa_s\{ \hat c_\alpha^\text{boot}(\eta_{s-1}) - (1-\alpha)\}, \quad s \geq 1, 
\end{equation}
where $(\kappa_s) \subset (0,1)$ is a deterministic sequence of step-sizes, e.g., $\kappa_s \propto (1+s)^{-\gamma}$ for $\gamma \in (0.5,1]$.  The stochastic approximation updates terminate when convergence is reached, e.g., if $|\eta_s - \eta_{s-1}|$ is smaller than some specified tolerance.  The output is a learning rate value, $\hat\eta_n$, depending on $T^n$, $\alpha$, and other inputs.  

The reader is sure to notice that the GPC algorithm is potentially computationally intensive.  In particular, it requires posterior computations on $B$ many bootstrap data sets of size $n$ at each iteration of the stochastic approximation update \eqref{eq:sa.update}.  In a misspecified Bayes setting, it is not out of the question that the model is sufficiently simple (e.g., conjugate priors) that posterior computations can be done in more-or-less closed form, in which case the GPC algorithm is relatively inexpensive.  In a Gibbs setting where the posterior is based on a loss function, it is unlikely that the posterior will be available in closed-form.  Regardless, if Markov chain Monte Carlo methods are required to compute the posterior, then the GPC algorithm is more expensive.  In our experience, however, it is not prohibitively expensive.  For one thing, there does not seem to be a benefit to having large $B$, so we have taken $B$ to be 200--300 in our applications.  Also, the posterior computations for the $B$ bootstrap data sets at a given $\eta$ can be done in parallel.  Finally, there does not seem to be any practical benefit to having a strict convergence criterion for stopping the updates in \eqref{eq:sa.update}, so capping the number of iterations to, say, 10--20 works fine in practice.  Apparently, just having $\eta$ in a neighborhood of the solution to ``$\hat c_\alpha^\text{boot}(\eta)=1-\alpha$'' is enough to achieve practical calibration.

\section{Numerical examples}
\label{S:examples}

\subsection{Quantile regression}
\label{ss:qr}

Suppose we observe data pairs $(X,Y) \sim P$, where $Y$ is the scalar response variable of primary interest and $X$ is a vector covariate.  Quantile regression \citep[e.g.,][]{koenker2005} models the $\tau^{\text{th}}$ conditional quantile of $Y$, given $X=x$, as a linear combination $\theta^\top f(x)$, for a known dictionary of functions $f(x) = (f_1(x), \ldots, f_J(x))$, with $\theta \in \RR^J$ an unknown vector of coefficients.  The quantity of interest is defined as the value $\theta^\star$ that minimizes the risk $R(\theta) = P\ell_\theta$, where the loss is the so-called {\em check loss function}
\[\ell_\theta(x,y) = \bigl(y - \theta^\top f(x) \bigr) \bigl(\tau - 1\{y < \theta^\top f(x)\} \bigr). \]
Given an iid sample $(X_i, Y_i)$ from $P$, with $i=1,\ldots,n$, and a prior distribution $\Pi$ for $\theta$, the Gibbs posterior for inference can be readily constructed via the formula \eqref{eq:gibbs}.  

As far as the Gibbs posterior's properties are concerned, the check loss function admits a $\theta$-derivative almost everywhere, given by
\[\dot{\ell}_\theta(x,y) = \Bigg\{\begin{matrix} (1-\tau) f(x), & \text{ if }y < \theta^\top f(x)\\
-\tau f(x), & \text{ if }y > \theta^\top f(x)
\end{matrix};\]
and twice-differentiable risk function with second derivative
\[\ddot{R}(\theta) = \int f(x)f(x)^\top p_x(\theta^\top f(x))P(dx)\]
where $p_x(y)$ denotes the conditional density of $Y$, given $X=x$.  Therefore, according to the theory presented in Section~\ref{SS:bvm}, this implies the Gibbs posterior distribution has a root-$n$ concentration rate at $\theta^\star$ and is asymptotically normal. 

For a concrete example, let $Y_i$ given $X_i = x_i$ be normally distributed with standard deviation $2$ and with median ($\tau = 0.5$ quantile) equal to $\theta_0 + \theta_1 x_i$ and where $X_i + 2 \sim \chisq(4)$, independent, for $i=1,\ldots,n$, with $\theta^\star = (2,1)$.  For this illustration, we take a flat prior for $\theta$.  In this case, $\Sigma_{\theta^\star} \approx 1.25 V_{\theta^\star}^{-1}$, which implies the Gibbs posterior can be asymptotically calibrated by taking $\eta^\star = 1.25^{-1/3} \approx 0.93$.  The goal here is to investigate the performance of the GPC algorithm, to see if it can effectively calibrate the Gibbs posterior.  To check this, we simulated 400 data sets of size $n=50$ from the aforementioned joint distribution and, for each data set, ran the GPC algorithm and then extracted the corresponding 95\% Gibbs posterior credible region for $\theta$.  In these simulations, the marginal coverage for $\theta_0$ and $\theta_1$ was 92.5\% and 94.5\%, respectively, both within a tolerable range of the target 95\%.  Here, the average learning rate selected by the GPC algorithm was 0.99, with standard deviation 0.19, which is in a neighborhood the asymptotic oracle value $\eta^\star \approx 0.93$.  Similar results were obtained with $n=400$.


If the GPC-calibrated Gibbs posterior produces reliable inferences for $\theta$ then one would expect it behaves similarly to other reasonable methods, e.g., the bootstrap.  For instance, $95\%$ Gibbs posterior joint credible sets for $\theta$ should be close to $95\%$ confidence sets based on the bootstrapped M-estimator.  For a specific comparison, consider the Gibbs posterior distribution of the elliptical transform $g(\theta) = (\theta - \bar \theta)^\top \Psi^{-1}(\theta - \bar \theta)$, where $\bar \theta$ and $\Psi$ are the mean and covariance matrix obtained from the Gibbs posterior sample, respectively.  The set $\{ \vartheta: g(\vartheta) \leq z_{0.95}^\text{\sc e}\}$, 
where $z_{0.95}^\text{\sc e}$ is the $0.95$ quantile of the marginal Gibbs posterior distribution of $g(\theta)$, defines a $95\%$ elliptical credible region for $\theta$.  This is the same shape as the asymptotic credible region, but its justification does not depend on any asymptotic result.  We also consider the highest posterior density credible region, defined by the set of posterior samples $\{\vartheta: e^{-\omega R_n(\vartheta)} \leq z_{0.95}^\text{\sc h}\}$, where $z_{0.95}^\text{\sc h}$ is the $0.95$ quantile of the marginal Gibbs posterior distribution of $e^{-\omega R_n(\theta)}$.  This region need not be elliptical, but we would expect it to be roughly elliptically shaped, at least for large $n$.  Figure~\ref{fig:qr} compares these two Gibbs posterior credible regions to elliptical $95\%$ confidence regions based on the bootstrapped M-estimator for sample sizes of $n \in \{50,400\}$.  All three regions are similar in shape and orientation, but the bootstrap-based region is slightly larger in both instances, and the highest posterior density region indeed looks more elliptical at the larger sample size.

\begin{figure}
\centering
\subfigure[$n=50$, $\hat\eta \approx 0.96$]{\scalebox{0.4}{\includegraphics{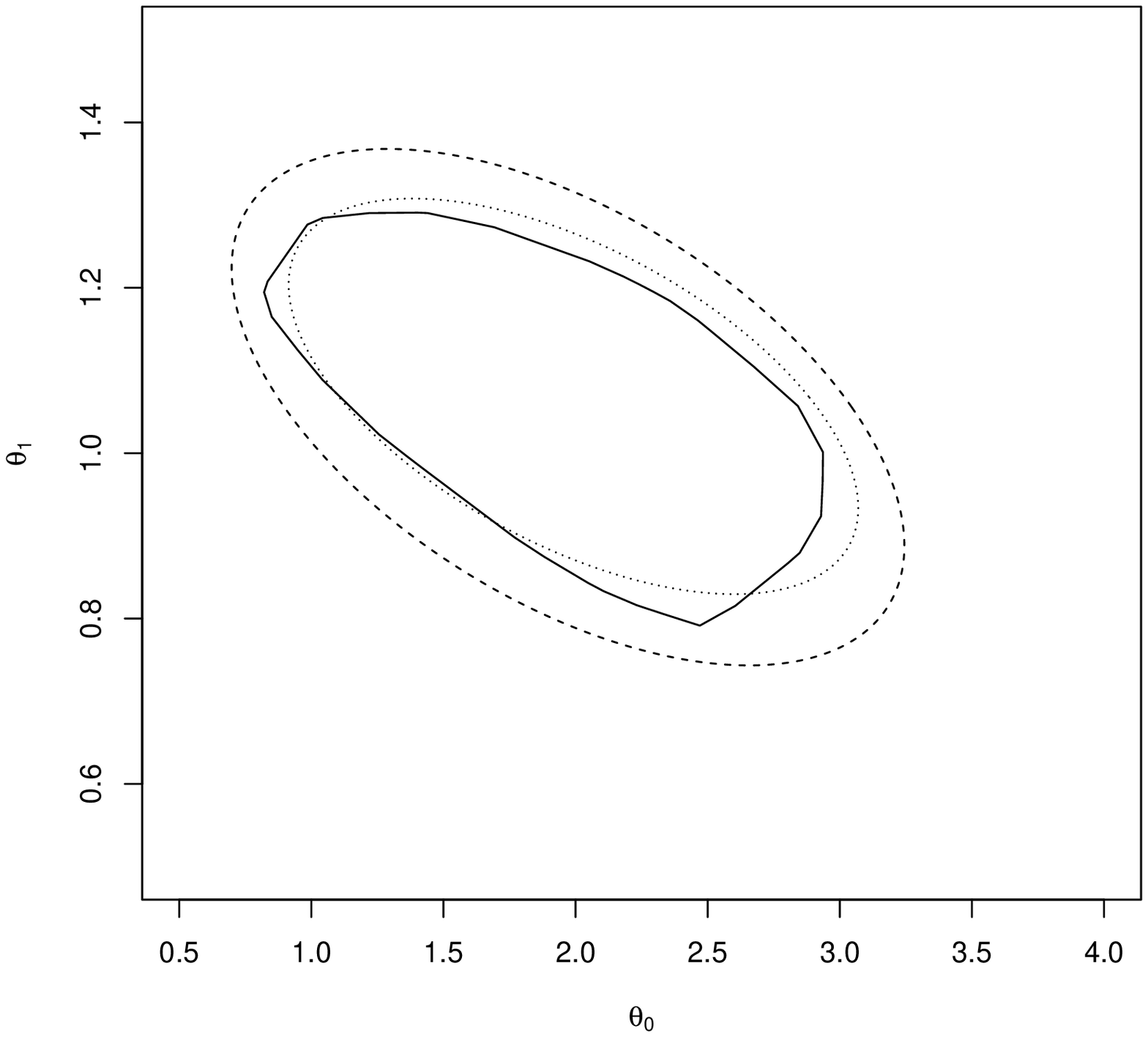}}}
\subfigure[$n=400$, $\hat\eta \approx 0.90$]{\scalebox{0.4}{
\includegraphics{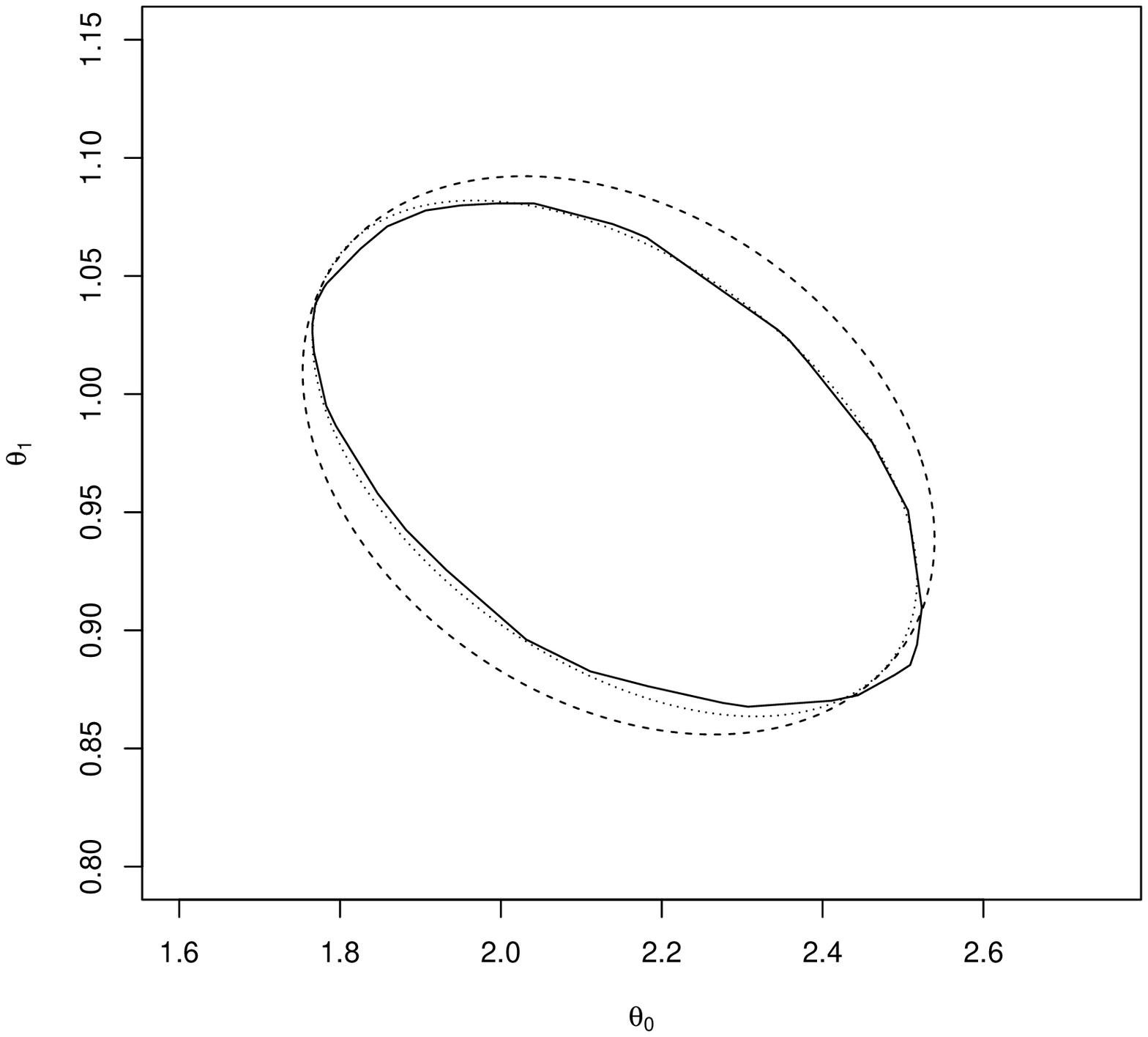}}}
\caption{95\% highest posterior density level set of the Gibbs posterior calibrated by the GPC algorithm (solid), $95\%$ elliptical credible region (dotted), and $95\%$ elliptical confidence region based on the bootstrapped M-estimator (dashed) for the quantile regression example in Section~\ref{ss:qr}.}
\label{fig:qr}
\end{figure}

\subsection{Classification}
\label{ss:svm}

Consider response-predictor data pair $(X,Y) \sim P$, where $X\in\RR^r$ is a continuous predictor and $Y\in\{-1,+1\}$ is a binary response---a class or label.  The classification boils down to learning the relationship between $X$ and $Y$, i.e., what values of $X$ tend to correspond to $Y=+1$ and vice versa.  This is typically carried out through specification of a {\em classifier}, a function that maps the $X$-space to $\{-1,+1\}$, often depending on a linear combination $\theta^\top f(x)$, where $\theta$ and $f(x)$ are as in Section~\ref{ss:qr}.  The unknown $\theta$ is linked to the data $(X,Y)$ through a choice of loss function.  A common choice is the 0--1 loss, 
\[ \ell_\theta(x,y) = 1 - y \, \text{sign}\{\theta^\top f(x)\}. \]
An advantage to this is interpretation: the expected loss is 
\[ R(\theta) = P[Y \neq \text{sign}\{\theta^\top f(X)\} ], \]
so the risk minimizer, $\theta^\star$, determines the classifier with smallest missclassification probability.  A disadvantage is that the discontinuity makes optimization of the empirical risk a challenging computational problem.  To remedy this, smooth versions of the 0--1 loss can be considered.  For example, the so-called {\em hinge loss} is given 
\[ \ell_\theta(x,y) = \max\{ 0, 1 - y \cdot \theta^\top f(x) \}. \]
This loss function is continuous and leads to a more manageable computational problem, which at least partially explains the popularity of support vector machines and maximum-margin classifiers.  In any case, once we have iid data $(X_i,Y_i)$ from $P$, have chosen a loss function, and specified a prior, the Gibbs posterior distribution for $\theta$ obtains as in \eqref{eq:gibbs}.  In what follows, we will focus on the hinge loss.  

As in Section~\ref{ss:qr}, we can ask what properties the the Gibbs posterior affords.  The hinge loss is continuous and almost everywhere $\theta$-differentiable, with derivative 
\[\dot{\ell}_\theta(x,y) = \begin{cases} -y f(x) & \text{if $1-y \cdot \theta^\top f(x) > 0$} \\ 0 & \text{otherwise}. \end{cases} \]
Moreover, the risk function is given by 
\begin{align*}
R(\theta) & = \int_{\{x:1-\theta^\top f(x)>0\}} \{1-\theta^\top f(x)\} \, m(x) \, P(dx) \\
& \qquad + \int_{\{x:1+\theta^\top f(x)>0\}} \{1+\theta^\top f(x)\} \, \{1-m(x)\} \, P(dx), 
\end{align*}
where $m(x) = P(Y = +1 \mid X=x)$ is the conditional probability function determined by $P$.  $R(\theta)$ admits two $\theta-$derivatives, which may be computed by the Leibniz integral rule, but we omit the (complicated) form of this function.  The point is that the existence of these derivatives implies both the M-estimator and Gibbs posterior for $\theta$ corresponding to the hinge loss are asymptotically normally distributed.


\ifthenelse{1=1}{}{
Then, three common classifiers are:
\begin{itemize}
    \item Parametric classifiers.  Let $p(x):=P(Y=1|X=x)$ and assume this is a function of $\theta^\top x$; for example, choose $p(x) = \exp(\theta^\top x)[1+\exp(\theta^\top x)]^{-1}$, which corresponds to logistic regression.  The estimated classifier equals $\hat\theta^\top  x$ where $\hat\theta$ is the MLE, the empirical risk minimizer of the log loss $\ell_\theta(x,y) = [p(x)]^{(1+y)/2}[1-p(x)]^{(1+y)/2}$.
    \item 0-1 classifiers.  These are non-parametric versions of the above strategy.  Let $\ell_\theta(x,y) = 1-y\cdot\text{sign}(\theta^\top x)$.  The corresponding risk is minimized at the value $\theta$ solving $p(x) = 1/2$.  
    \item Support-vector machines (SVM).  The SVM chooses the linear function $\theta^\top x$ best separating $x$ values of different classes by minimizing the penalized hinge loss $\ell_\theta(x,y) = \max(0, 1-y \cdot\theta^\top  x) + \nu^{-1}\|\theta\|_1$ where $\nu>0$ is a tuning parameter.    
\end{itemize}}

For a concrete example of a Gibbs posterior for classification with the hinge loss let $m(x) = F({\theta^\star}^\top  x)$, where $F$ denotes the distribution function of a Student's $t$ random variable with $3$ degrees of freedom, $\theta^\star = (1,-1)^\top$, and $X\sim \nm(1,1)$.  We investigate the behavior of the Gibbs posterior calibrated by the GPC algorithm targeting $95\%$ credible intervals for $\theta_1^\star$ in a short simulation of 400 replications for sample sizes $n=400$ and $1000$.  At sample size $400$ the average learning rate selected by GPC was about $0.77$ with standard deviation $0.05$.  For $n=1000$ the learning rate selected by GPC was a bit smaller, averaging $0.65$.  For both simulations the GPC-calibrated $95\%$ Gibbs posterior credible intervals were conservative, with coverage of about $99\%$ for $\theta_1$.  

Figure~\ref{fig:svm} displays the same three types of joint credible/confidence regions for $\theta^\star$, the hinge risk minimizer, as in Figure~\ref{fig:qr}.  For both moderate and large sample sizes the Gibbs highest posterior density credible region is very similar to the elliptical region, and , hence, similar to the shape of the asymptotic credible region.  Both Gibbs posterior credible regions contain the bootstrap-based confidence region, which reflects the over-coverage observed in the simulation experiment. 

\begin{figure}
\centering
\subfigure[$n=400$, $\hat\eta \approx 0.70$]{\scalebox{0.4}{\includegraphics{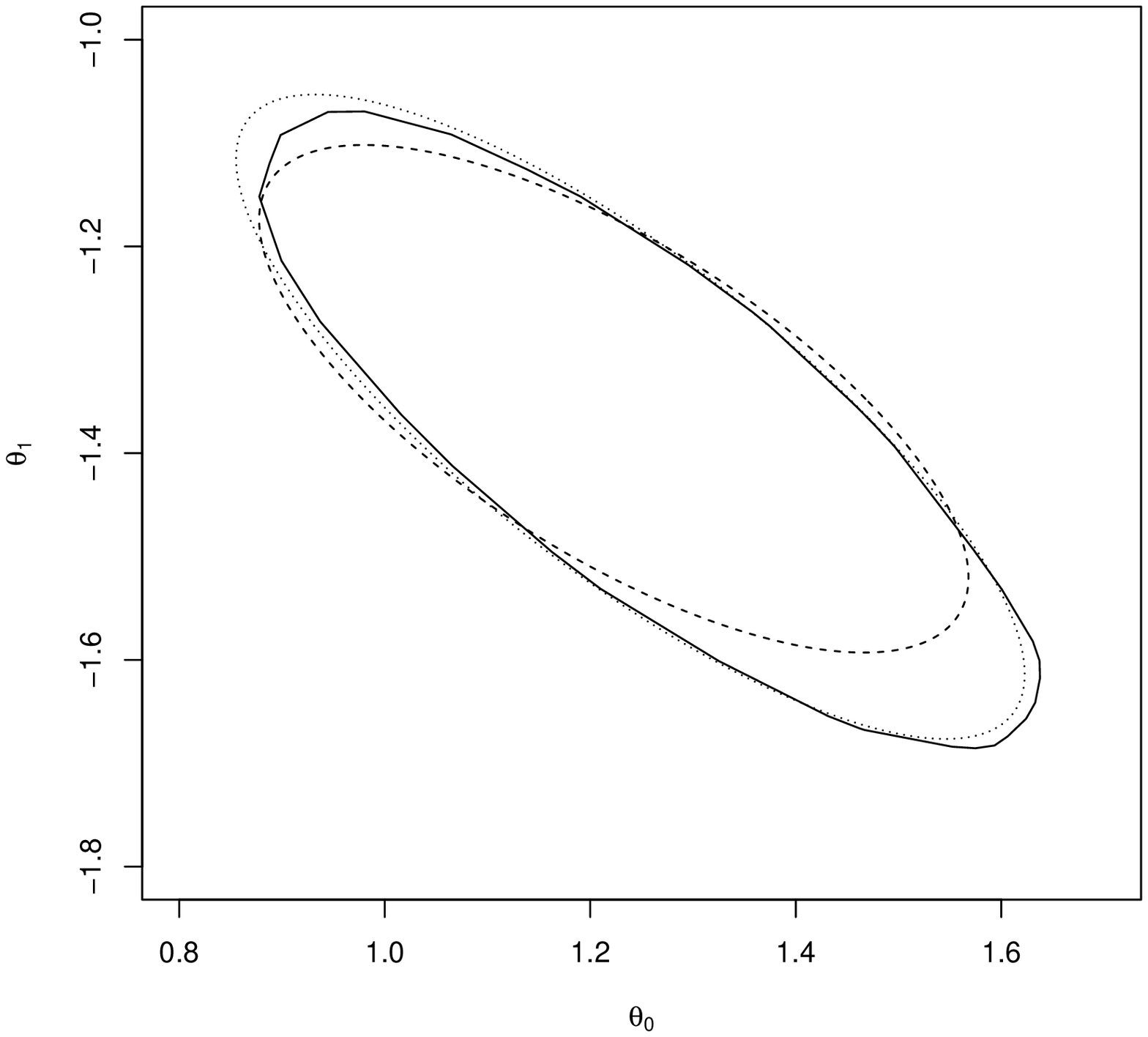}}}
\subfigure[$n=1000$, $\hat\eta \approx 0.58$]{\scalebox{0.4}{
\includegraphics{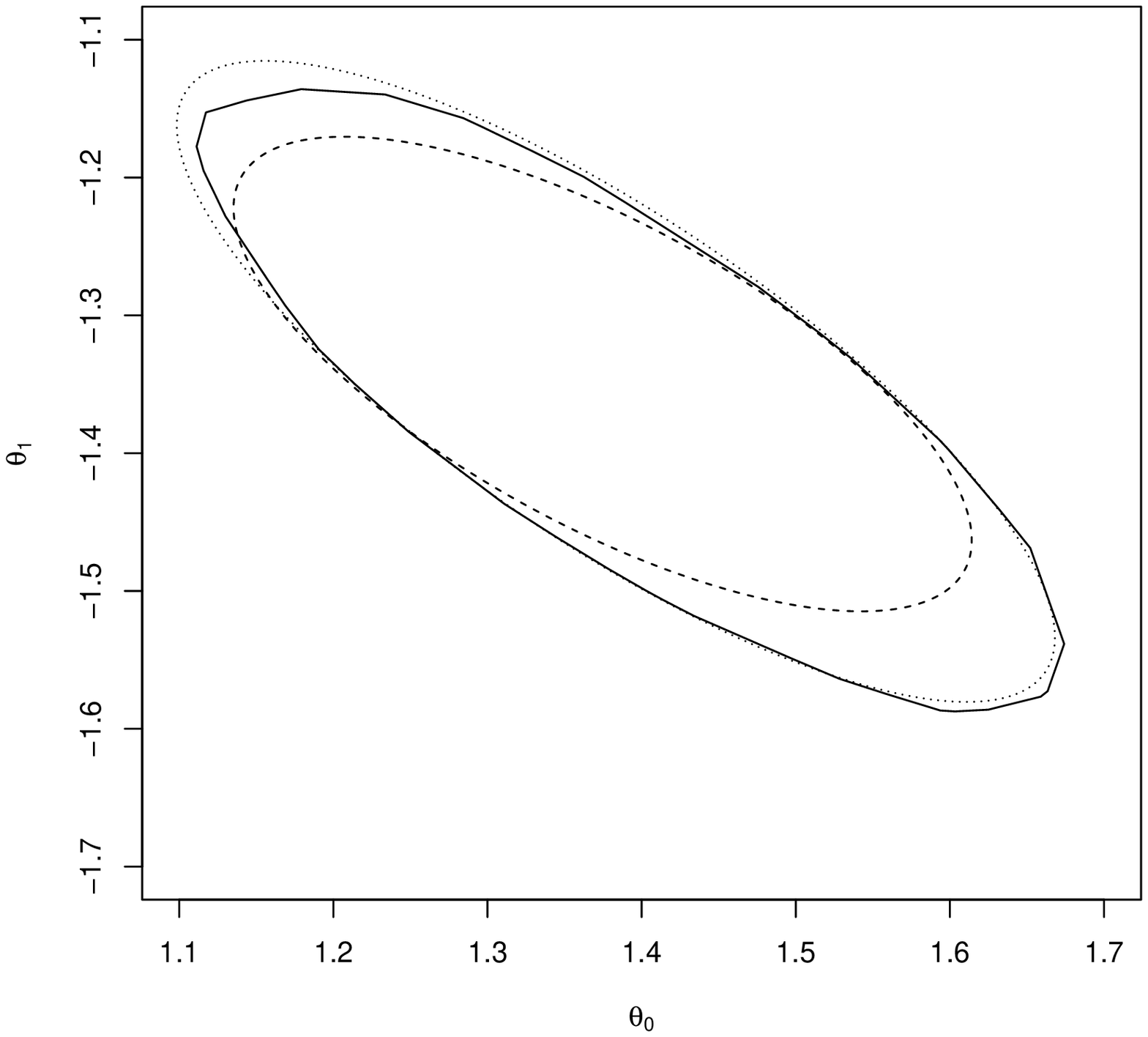}}}
\caption{95\% highest posterior density level set of the Gibbs posterior calibrated by the GPC algorithm (solid), $95\%$ elliptical credible region (dotted), and $95\%$ elliptical confidence region based on the bootstrapped M-estimator (dashed) for the classification regression example in Section~\ref{ss:svm}.} 
\label{fig:svm}
\end{figure}

\subsection{Non-linear regression}
\label{ss:NPreg}

So far, all our examples have considered finite-dimensional inference problems, but Gibbs posteriors can be used for inference on high- or even infinite-dimensional parameters as well.  Mean regression is a common application of high-dimensional inference and one setting in which Gibbs posteriors have already been studied; see, for example, \citet{gibbs.general}.  Let $(X,Y) \sim P$ and consider the loss function $\ell_\theta(x,y) = \{y - \theta(x)\}^2$, for $\theta$ a generic smooth function.  If $P$ admits a finite second moment, then it is easy to show that the risk minimizer, $\theta^\star$, exists; if the function class is sufficiently broad, then $\theta^\star(x)$ equals the conditional expectation of $Y$, given $X=x$, under $P$.  In any case, with iid data $(X_i, Y_i)$ and a suitable prior $\Pi$ on $\theta$, we can construct a Gibbs posterior distribution for inference on the risk minimizer.  

For smooth functions, the so-called random series priors \citep{shen.ghosal.2015} are quite convenient.  Parametrize $\theta$ by a linear combination of a chosen set of basis functions $\theta(x) = \phi^\top f(x)$ where $\phi = (\phi_1, \phi_2, \ldots, \phi_J)\in\RR^J$, and $f(x) = (f_1(x), f_2(x), \ldots, f_J(x))^\top$ denotes the first $J$ basis functions from, say, a Fourier, spline, or polynomial basis.  Then, a prior on $\theta$ is induced by a hierarchical prior on $J$ and on $\phi$, given $J=j$.  Common choices include a Poisson prior on $J$ and normal conditional priors on the coefficients in $\phi$.  Having a prior distribution on the number of basis functions makes the model flexible and adaptive to functions $\theta$ of different levels of smoothness.  In practice, the posterior may perform well for a fixed $J$ and a prior on the coefficient vector only.      

A useful feature of a Gibbs posterior for $\theta$ is a $100(1-\alpha)\%$ uniform credible band, a sup-norm ball of functions $\theta$ having $1-\alpha$ Gibbs posterior probability.  This can be used as a confidence band, i.e., as a set of functions with coverage probability $1-\alpha$, but this calibration would not be automatic.  Fortunately, the GPC algorithm can used to calibrate these posterior credible regions, even in this infinite-dimensional setting.

For an example of applying the Gibbs posterior along with GPC to non-linear regression, consider $X \sim \unif(0,1)$ and $(Y \mid X=x) \sim \nm(\theta^\star(x),0.2^2)$ where $\theta^\star(x) = 20x^3 - 34x^2 + 15.2 x - 1.2$.  We used independent, diffuse normal priors for the coefficients of a cubic polynomial basis, so the number of basis functions is fixed at $J=4$.  Figure~\ref{fig:GP} displays $95\%$ Gibbs posterior uniform credible bands for $\theta$, based on $n=100$, and using the GPC-selected learning rate.  In this case, the credible band contains $\theta^\star$. 

\begin{figure}[t]
\centering
\scalebox{0.5}{\includegraphics{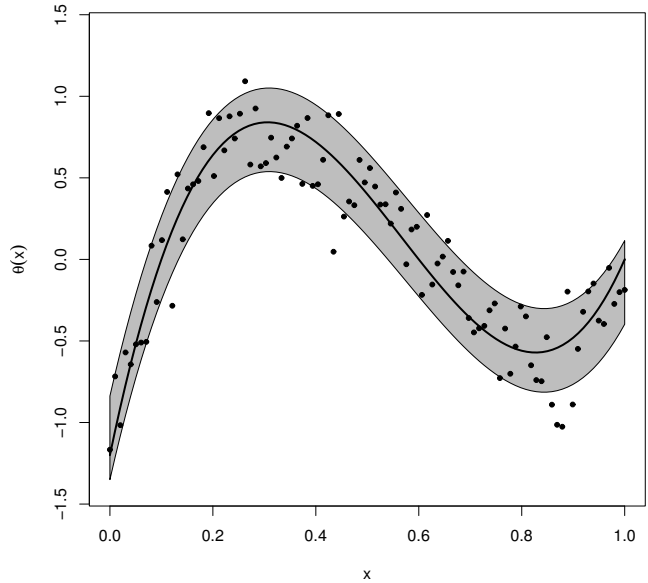}}
\caption{95\% uniform Gibbs posterior credible bands with learning rate selected by GPC for the non-linear regression example in Section~\ref{ss:NPreg}.} 
\label{fig:GP}
\end{figure}

\section{Further details}
\label{S:more}

\subsection{Things we did not discuss}
\label{SS:others}

Our coverage of the relevant results in the Gibbs or generalized Bayes posterior inference was necessarily limited.  So there are lots of interesting and important ideas and results that we did not discuss.  For the sake of being semi-complete in our survey of the relevant literature, here we briefly mention a few of these other problems and directions.  

\begin{itemize}
\item Gibbs is not the only alternative to Bayes.  Recall that Gibbs is closely related to M-estimation and empirical risk minimization, where the parameter of interest may be defined as the minimizer of an expectation/risk function and point estimates are derived by minimizing the empirical analog of the risk.  A closely related technique is Z-estimation defined by moment conditions in which one or more moments/expectations are exactly zero as functions of $\theta$ at $\theta = \theta^\star$.  \citet{chernozhukov.hong.2003}  develop a quasi-posterior distribution for models based on moment conditions taking similar form to a Gibbs posterior with an empirical risk function that is defined as a quadratic form based on the moment conditions.  An alternative approach for handling models defined by moment conditions is the exponentially-tilted empirical likelihood (ETEL) posterior; see, for example, \citet{chib.etal.etel}.  That approach utilizes an empirical likelihood in place of the usual parametric likelihood, restricted to distributions satisfying the moment conditions, and combined with a non-parametric prior favoring the empirical likelihood.  A variation of this approach is the penalized ETEL appearing in, for example, \citet{petel}, in the context of quantile regression.  The PETEL posterior combines the ETEL approach with a penalty term taking the form of the relevant empirical risk function that forms the basis of the Gibbs posterior.   

\item An advantage of the Bayesian framework is that once the posterior is in hand, answers to any relevant question can be derived from it.  One of these relevant questions concerns prediction of future observations.  When the model is incorrectly- or under-specified, this advantage disappears.  In particular, the standard/naive construction of a predictive distribution need not have good properties even if the posterior does.  \citet{wu.martin.gprc} considered the prediction problem, proposed a generalized predictive distribution, and developed a variation on the GPC algorithm described above that would ensure the prediction intervals derived from it would be calibrated in a frequentist sense. 

\item When the loss function is rough, the empirical risk function $R_n$ tends to be rough too.  From the empirical risk minimization point of view, this roughness can make optimization a challenge.  One option is to suitably smooth the rough objective function so that the optimization problem remains (largely) unchanged.  While smoothing may not significantly affect the estimation problem, it can create challenges with inference.  From a Gibbsian perspective, the rough empirical risk can create problems for designing an efficient posterior sampling algorithm, so here too some smoothing might be desirable.  However, the type of smoothing that leads to simpler optimization may not lead to efficient posterior sampling, so perhaps some different considerations are required.  Alternatively, one might consider a {\em variational approximation} to the Gibbs posterior based on the rough empirical risk function; see, e.g., \citet{alquier.etal.2016} and \citet{alquier.ridgway.2020}. This would have simple posterior computations by construction, but might not be as accurate of an approximation of the original Gibbs posterior compared to one that directly and appropriately smooths the empirical risk.  
\end{itemize}

\subsection{Open problems}

Generalized Bayesian inference has been an active area of research in recent years.  While lots of exciting new developments have been made, there are still a number of interesting and important questions that remain unanswered.  We take the opportunity here to list just a few open problems.  This is just a biased selection, far from an exhaustive list.  

\begin{itemize}
\item The GPC algorithm described in Section~\ref{S:learning.rate} above has been shown to have strong empirical performance in a fairly wide range of applications.  However, there is still no rigorous theoretical support to back this up.  The challenge is that there a lot of moving parts: posterior computations via Monte Carlo, bootstrap, and stochastic approximation.  All three of these methods individually are theoretically sound, but GPC applies them simultaneously, which markedly complicates the analysis.  

\item Here the learning rate appeared as a power in the pseudo-likelihood expression.  However, there may be other ways in which a ``learning rate'' parameter might appear in a generalized posterior construction.  For example, composite likelihoods and the corresponding posterior distributions \citep{pauli.etal.2011} often involve at least one adjustment factor that could be tuned via GPC.  Similarly, variational approximations are known to under-estimate the posterior spread \citep{blei.etal.vb.2017}, so one could introduce an additional adjustment factor that, again, can be tuned using GPC.  So we see the idea behind the GPC algorithm as a general strategy that can be applied beyond the Gibbs posterior applications discussed here.

\item To our knowledge, learning rate adjustment via the GPC algorithm has only been investigated in relatively low-dimensional problems.  Our expectation is that GPC's performance will deteriorate as the dimension of $\theta$ increases, but it is currently unknown how quickly this deterioration would happen.  Can the algorithm be modified to scale more efficiently with dimension, or is an entirely new algorithm needed? 

\item As we argued above, the introduction of the scalar learning rate was a simple consequence of the loss function's scale relative to the prior being ambiguous.  But having only a scalar learning rate to tune limits our ability to adjust the shape of the Gibbs posterior.  Other kinds of adjustments are possible, e.g., the sandwich likelihood in \citet{sriram2015}.  What about other more general ``learning rate structures'' that would allow for simultaneous adjustment of both the shape and spread of the Gibbs posterior contours? 

\item The theoretical results presented here, as well as those in \citet{gibbs.general}, \citet{grunwald.mehta.rates}, and elsewhere, focus exclusively on estimation-related question such as consistency and concentration rates.  When $\theta$ is high-dimensional, it is common for there to be an underlying low-dimensional structure that is of practical relevance, so there is a question of whether the Gibbs posterior would be able to learn that structure.  For example, consider a classification problem that involves a high-dimensional feature $x$.  In such cases, one might be willing to believe that only a relatively small fraction of all the features should affect the classifier, so a sparsity-encouraging prior might be used.  But which features are important or active in the optimal classifier is unknown, and a natural question is if the Gibbs posterior can identify these.  This specific question was addressed in \citet{jiang.tanner.2008}, and another similar result in a different context we presented recently in \citet{wang.martin.levy}.  To our knowledge, however, there have been no general investigations into Gibbs posterior structure learning.  

\item We have focused exclusively here on cases where the data $T^n = (T_1,\ldots,T_n)$ are iid from a common distribution $P$.  From here, an extension to a case where data $T_i$ are independent but not iid, having distinct distributions $P_i$ for $i=1,\ldots,n$, would not be out of reach.  The case, however, where the $T_i$'s are {\em dependent} has, to our knowledge, not been given much attention in the literature.  That the exponent in the definition of $\pi_n^{(\eta)}$ in, say, \eqref{eq:gibbs.density} involves a sum of individual negative loss terms seems uniquely suited for independent data, so all of what has been presented here would need to be reworked.  

\end{itemize}

\section{Conclusion}
\label{S:discuss}

This paper considered the problem in which the quantity of interest is defined, not as a parameter in a statistical model for the data-generating process, but as something that is, or at least can be, expressed as a minimizer of a suitable expected loss function.  It is often the case that quantities having a real-world interpretation, beyond a statistical model, can be expressed in this way, e.g., quantiles and moments.  More generally, the quantities of interest in machine learning applications can often be expressed as risk minimizers.  Of course, regardless of how the quantity of interest is most naturally defined, it would be possible to formulate a statistical model, recast the quantity of interest in terms of the model parameters, and proceed with Bayesian inference as usual.  Here we argued, first, that there are good reasons for not going this indirect route through a statistical model:
\begin{itemize}
\item no risk of model misspecification bias;
\vspace{-2mm}
\item no need to over-complicate matters by introducing nonparametric models;
\vspace{-2mm}
\item no need to deal with prior specification and posterior computations pertaining to nuisance parameters and the associated challenges with marginalization \citep{fraser2011, martin.nonadditive}.
\end{itemize} 
Second, we argued that a direct approach can be carried out using a Gibbs posterior.  Interpretation of the Gibbs posterior is different from that of a Bayesian posterior, but it can still be used for making inference, and it shares many of the familiar asymptotic convergence properties of the Bayes posterior.  A key point is that the Gibbs posterior is not automatically calibrated in a frequentist sense.  Calibration only holds for Bayesian posteriors in correctly specified models, but since ``All models are wrong...''~this Bayesian result does not provide much comfort.  We argued that calibration can be achieved, just not automatically---we need to carefully tune the learning rate parameter.  With the Gibbs posterior's desirable convergence properties, together with a suitable, data-driven learning rate selection procedure, this appears to us to be a powerful framework, fundamental to what could be described as Bayesian statistical learning.  

This, of course, was a biased survey of recent developments falling under the umbrella of generalized Bayesian inference.  We briefly mentioned a few other ideas and approaches in Section~\ref{SS:others} but that definitely does not do these developments justice.  While there are some technical differences between the approach advocated for here and those advocated by others, we want to end this discussion by highlight what they all have in common.  Wasserman's quote from Section~\ref{S:intro} is right: Bayesian inference is too rigid in its insistence on being able to answer all relevant questions about the data-generating process in one stroke.  But let's not throw out the baby with the bath water.  That is, there is no need to abandon hope of principled, probabilistic inference on interest parameters to get the flexibility Wasserman is looking for.  We just need to be more strategic/direct with our posterior construction.  The developments here and elsewhere in the generalized Bayes literature are in this vein, and we are excited to see where this goes.

\section*{Acknowledgments}

This work is partially supported by the U.~S.~National Science Foundation, grant number SES--2051225.  The authors also thank the editors of this {\em Handbook}, Alastair Young in particular, for the invitation to make a contribution.

\appendix

\section{Proofs}
\label{S:proofs}

\subsection{Proof of Theorem~\ref{thm:cons}}
\label{proof:thm1}

The proofs of both Theorems~\ref{thm:cons} and \ref{thm:conc} share a similar strategy.  Start by redefining the Gibbs posterior distribution as the ratio
\[ \Pi_n^{(\eta)}(A) = \frac{N_n^{(\eta)}(A)}{D_n^{(\eta)}}, \quad A \subseteq \Theta, \]
where the numerator and denominator, respectively, are 
\begin{align*}
N_n^{(\eta)}(A) & = \int_A e^{-\eta\{R_n(\theta) - R_n(\theta^\star)\}} \, \Pi(d\theta) \\
D_n^{(\eta)} & = \int_\Theta e^{-\eta\{R_n(\theta) - R_n(\theta^\star)\}} \, \Pi(d\theta). 
\end{align*}
Gibbs posterior consistency requires $\Pi_n^{(\eta)}(A_\eps) \to 0$, for $A_\eps = \{\theta: d(\theta, \theta^\star) > \eps\}$, for any $\eps > 0$.  We proceed by showing that (a)~the numerator is vanishing and (b)~the denominator is not any smaller than the bound on the numerator, both as $n \to \infty$.  We start with Lemma~\ref{lem:den_cons} below that bounds the denominator.  


\begin{lem}
\label{lem:den_cons}
If \eqref{eq:prior_cond} holds, then $P\{D_n^{(\eta)} > e^{-n \eta \delta}\} \to 1$ as $n \to \infty$ for any $\delta > 0$.
\end{lem}

\begin{proof}
Begin by lower-bounding $e^{n\eta\delta}D_n^{(\eta)}$ by restricting the domain of integration in the definition of $D_n^{(\eta)}$:
\begin{align*}
e^{n\eta \delta}D_n^{(\eta)} 
    & = \int e^{-\eta n \{R_n(\theta)-R_n(\theta^\star) - \delta\}} \, \Pi(d\theta)\\
    & \geq \int_{\{\theta: R(\theta) - R(\theta^\star) \leq \delta / 2\}} e^{\eta n \{\delta - R_n(\theta)+R_n(\theta^\star)\}} \, \Pi(d\theta). 
\end{align*}
The law of large numbers implies that $R_n(\theta) - R_n(\theta^\star)$ converges $P$-almost surely to $R(\theta) - R(\theta^\star)$, pointwise in $\theta$.  So, for $\theta$ in the above range of integration, we get 
\[ \limsup_{n \to \infty} \{ R_n(\theta) - R_n(\theta^\star)\} \leq \tfrac{\delta}{2}, \quad \text{$P$-almost surely}. \]
Then the integrand in the lower bound for $e^{n\delta} D_n^{(\eta)}$ is converging to $\infty$ pointwise in $\theta$, $P$-almost surely.  Then Fatou's lemma and the condition \eqref{eq:prior_cond} on the prior mass assigned to risk neighborhoods of $\theta^\star$ together imply that 
\[ \liminf_{n \to \infty} e^{n\delta} D_n^{(\eta)} = \infty, \quad \text{$P$-almost surely}, \]
and, from this, (an even stronger version of) the claim follows.  
\end{proof}

In contrast to the the denominator, which is controlled by properties of the prior, the behavior of the numerator is largely determined by properties of the loss function.  

\begin{lem}
\label{lem:num_cons}
If \eqref{eq:num_cons} holds, then $P\{N_n^{(\eta)}(A_\eps) \leq e^{-\eta n c}) \to 1$ as $n \to \infty$, for a constant $c > 0$ depending only on $\eps$. 
\end{lem}

\begin{proof}
The risk difference can clearly be rewritten as 
\[ R_n(\theta) - R_n(\theta^\star) = \{R_n(\theta) - R(\theta)\} + \{R(\theta) - R(\theta^\star)\} + \{R(\theta^\star) - R_n(\theta^\star)\}. \]
By \eqref{eq:num_cons_a}, the first term is $o_P(1)$ uniformly in $\theta$, by \eqref{eq:num_cons_b} the second term is lower-bounded, uniformly on $A_\eps$, by a constant $\xi = \xi(\eps) > 0$, and by the law of large numbers the third term is $o_P(1)$ and does not depend on $\theta$.  Therefore, 
\begin{align*}
N_n^{(\eta)}(A_\eps) & = \int_{A_\eps} e^{-\eta n\{R_n(\theta) - R_n(\theta^\star)\}} \, \Pi(d\theta) \\
& \leq \int_{A_\eps} e^{-\eta n \{o_P(1) + \xi\}} \, \Pi(d\theta).
\end{align*}
Since the $o_P(1)$ term vanishes uniformly in $\theta \in A_\eps$, the bracketed term in the exponent will eventually be bigger than, say, $\xi/2 > 0$.  Therefore, with $P$-probability converging to 1, we have that $N_n^{(\eta)}(A_\eps) \leq e^{-\eta n c}$ for some constant $c > 0$, as claimed.  
\end{proof}

Lemmas~\ref{lem:den_cons} and \ref{lem:num_cons} together imply Gibbs posterior consistency. Indeed, on a set with $P$-probability converging to 1, we have that $N_n^{(\eta)}(A_\eps)$ is exponentially small and $D_n^{(\eta)}$ is not exponentially small.  Putting these two results together gives 
\[ \Pi_n^{(\eta)}(A_\eps) = \frac{N_n^{(\eta)}(A_\eps)}{D_n^{(\eta)}} \leq e^{-\eta n(c - \delta)}. \]
The constant $c > 0$ is fixed, depends on $\eps$, but $\delta > 0$ is arbitrary.  So if we take $\delta < c$, then we can conclude that, on a set with $P$-probability converging to 1, $\Pi_n^{(\eta)}(A_\eps) \to 0$, which proves consistency.

\subsection{Proof of Theorem~\ref{thm:conc}}
\label{proof:thm2}

Our strategy for proving concentration we present in this section mirrors our proof of consistency above.  First, we express the Gibbs posterior probability of the complement of a shrinking neighborhood $A_n = \{\theta:d(\theta,\theta^\star) > M_n \eps_n\}$ as the ratio $N^{(\eta)}(A_n)/D_n^{(\eta)}$.  Then we show that conditions \eqref{eq:conc_num} and \eqref{eq:conc_den} imply that the numerator is small and the denominator is not too small such that the ratio is vanishing.  These bounds on the denominator and numerator are established in the two lemmas presented next. 


\begin{lem}
\label{lem:conc_den}
Define the mean and variance functions of the excess loss:
\begin{align*}
    & m(\theta, \theta^\star) = R(\theta) - R(\theta^\star), \,\text{and} \\
    & v(\theta, \theta^\star) = P(\ell_\theta-\ell_{\theta^\star})^2 - m(\theta,\theta^\star)^2.
\end{align*}
And, define the subset of the parameter space $\Theta_n:=\{\theta: m(\theta,\theta^\star)\vee v(\theta,\theta^\star)
\leq C \eps_n^\alpha\}$ for some $C>0$ and where $\eps_n$ and $(\alpha,\beta)$ are defined in Theorem~\ref{thm:conc} and \eqref{eq:conc_num}.  Then, 
\[ D_n^{(\eta)} \gtrsim \Pi(\Theta_n)e^{-2b_n\eta n\eps_n^\alpha}, \quad \text{with $P^n$-probability $\to 1$}, \]
for any positive sequence $b_n \to \infty$.
\end{lem}

\begin{proof}
Define a standardized version of the empirical risk difference, i.e., 
\[ Z_n(\theta) = \frac{\{nR_n(\theta) - nR_n(\theta^\star)\} - n m(\theta, \theta^\star)}{\{n v(\theta, \theta^\star)\}^{1/2}}. \]  Of course, $Z_n(\theta)$ depends (implicitly) on the data $U^n$.  Let 
\[ \Z_n = \{(\theta, U^n): |Z_n(\theta)| \geq (b_n n \eps_n^\alpha)^{1/2} \}. \]
Next, define the cross-sections 
\[ \Z_n(\theta) = \{U^n: (\theta, U^n) \in \Z_n\} \quad \text{and} \quad \Z_n(U^n) = \{\theta: (\theta, U^n) \in \Z_n\}. \]
For $\Theta_n$ as above, since 
\[ nR_n(\theta) - nR_n(\theta^\star) = n m(\theta, \theta^\star) + \{n v(\theta, \theta^\star)\}^{1/2} Z_n(\theta), \]
and $m$, $v$, and $Z_n$ are suitably bounded on $\Theta_n \cap \Z_n(U^n)^c$, we immediately get 
\[ D_n^{(\eta)} \geq \int_{\Theta_n \cap \Z_n(U^n)^c} e^{-\eta n m(\theta, \theta^\star) - \eta \{n v(\theta, \theta^\star)\}^{1/2} Z_n(\theta)} \, \Pi(d\theta) \geq e^{-2b_n \eta n \eps_n^\alpha} \Pi\{\Theta_n \cap \Z_n(U^n)^c\}. \]
From this lower bound, we get  
\begin{align*}
P^n\{D_n^{(\eta)} \leq \tfrac12 \Pi(\Theta_n) e^{-2b_n\eta n \eps_n^\alpha}\} & \leq P^n\bigl[ e^{-2b_n \eta n \eps_n^\alpha} \Pi\{\Theta_n \cap \Z_n(U^n)^c\} \leq \tfrac12 \Pi(\Theta_n) e^{-2b_n\eta n \eps_n^\alpha} \bigr] \\
& = P^n\bigl[ \Pi\{\Theta_n \cap \Z_n(U^n)\} \geq \tfrac12 \Pi(\Theta_n) \bigr] \\
& \leq \frac{2 P^n \Pi\{\Theta_n \cap \Z_n(U^n)\}}{\Pi(\Theta_n)},
\end{align*}
where the last line is by Markov's inequality.  We can then simplify the expectation in the upper bound displayed above using Fubini's theorem:
\begin{align*}
P^n \Pi\{\Theta_n \cap \Z_n(U^n)\} & = \int \int 1\{\theta \in \Theta_n \cap \Z_n(U^n)\} \, \Pi(d\theta) \, P^n(dU^n) \\
& = \int \int 1\{\theta \in \Theta_n\} \, 1\{\theta \in \Z_n(U^n)\} \, P^n(dU^n) \, \Pi(d\theta) \\
& = \int_{\Theta_n} P^n\{\Z_n(\theta)\} \, \Pi(d\theta).
\end{align*}
By Chebyshev's inequality, $P^n\{\Z_n(\theta)\} \leq (b_n n\eps_n^\alpha)^{-1}$, and hence  
\[ P^n\{D_n^{(\eta)} \leq \tfrac12 \Pi(\Theta_n) e^{-2b_n\eta n \eps_n^\alpha}\} \leq 2(b_n n\eps_n^\alpha)^{-1}. \]
Finally, since $\alpha \geq 2\beta$ implies $n\eps_n^\alpha \geq 1$ and $b_n \to \infty$, the upper bound is vanishing, which proves the claim. 
\end{proof}

Lemma~\ref{lem:conc_num} below is a strengthening of Lemma~\ref{lem:num_cons}, and is used to obtain an upper bound on the numerator $N_n(A_n)$ with probability approaching 1.  First, we need to define some notation: let $\mathbb{P}_n = n^{-1}\sum_{i=1}^n \delta_{T_i}$ denote the empirical distribution where $\delta_{t}$ is the Dirac point-mass at $t$.  Let $\mathbb{G}_nf = n^{1/2}(\mathbb{P}_nf - Pf)$ denote the empirical process.     

\begin{lem}
\label{lem:conc_num}
Under the conditions of Theorem~\ref{thm:conc}, 
\[ N_n(A_n) \lesssim (M_n \eps_n)^q e^{-\eta c M_n^\alpha n \eps_n^\alpha}, \quad \text{with $P^n$-probability $\to 1$}, \]
where $c > 0$ is a constant. 
\end{lem}

\begin{proof}
Start by expressing the empirical risk difference as 
\[ R_n(\theta^\star) - R_n(\theta) = \{R(\theta^\star) - R(\theta) - n^{-1/2}\mathbb{G}_n(\ell_\theta - \ell_{\theta^\star})\}. \]
Then the Gibbs posterior numerator can be written as
\begin{align*}
N_n^{(\eta)}(A_n) & = \sum_{t=1}^\infty \int_{t M_n \eps_n < d(\theta, \theta^\star) < (t+1) M_n \eps_n} e^{-\eta n \{R_n(\theta) - R_n(\theta^\star)\}} \, \Pi(d\theta) \\
& = \sum_{t=1}^\infty \int_{t M_n \eps_n < d(\theta, \theta^\star) < (t+1) M_n \eps_n} e^{-\eta n \{R(\theta^\star) - R(\theta) - n^{-1/2}\mathbb{G}_n(\ell_\theta - \ell_{\theta^\star})\}} \, \Pi(d\theta).
\end{align*}
By \eqref{eq:conc_num1}, the right-hand side above can be upper bounded as 
\[ N_n^{(\eta)}(A_n) \leq \sum_{t=1}^\infty e^{-\eta C t^\alpha M_n^\alpha n \eps_n^\alpha} \int_{d(\theta, \theta^\star) < (t+1) M_n \eps_n} e^{\eta n^{1/2} |\GG_n(\ell_\theta - \ell_{\theta^\star})|} \, \Pi(d\theta), \]
where $C > 0$ is the constant hidden in ``$\gtrsim$'' in \eqref{eq:conc_num1}.  Next, it follows immediately from \eqref{eq:conc_num2} and Markov's inequality that, for any $a_n \to \infty$, 
\[ \sup_{d(\theta, \theta^\star) < (t+1)M_n \eps_n} |\GG_n(\ell_\theta - \ell_{\theta^\star})| \lesssim a_n \{(t+1) M_n \eps_n\}^\beta, \]
with $P^n$-probability converging to 1 as $n \to \infty$.  This uniform bound on the integrand above leads to 
\begin{align*}
N_n^{(\eta)}(A_n) & \leq \sum_{t=1}^\infty e^{-\eta C t^\alpha M_n^\alpha n \eps_n^\alpha + \eta n^{1/2} a_n \{(t+1) M_n \eps_n\}^\beta} \, \Pi\{d(\theta, \theta^\star) < (t+1) M_n \eps_n\} \\
& \leq (M_n \eps_n)^q \sum_{t=1}^\infty (t+1)^q e^{-\eta C t^\alpha M_n^\alpha n \eps_n^\alpha + \eta n^{1/2} a_n \{(t+1) M_n \eps_n\}^\beta}, 
\end{align*}
where the second inequality follows by the condition on the prior distribution $\Pi$.  Next, the term in the exponent can be upper bounded by  
\[ -\eta C t^\alpha M_n^\alpha n \eps_n^\alpha [ 1 - C' a_n \{(t+1) M_n \eps_n\}^{\beta-\alpha} n^{-1/2}]. \]
Since $\eps_n^{\beta-\alpha} = n^{1/2}$ and we are free to choose $a_n \ll M_n^{\alpha-\beta}$, it follows that the term in square brackets is bigger than, say, $\frac12$ for all sufficiently large $n$.  Therefore,  
\[ N_n^{(\eta)}(A_n) \lesssim (M_n \eps_n)^q \sum_{t=1}^\infty t^q e^{-\eta (C/2) t^\alpha M_n^\alpha n \eps_n^\alpha}, \quad \text{with $P^n$-probability $\to 1$}. \]
The summation is of the order $e^{-\eta c M_n^\alpha n \eps_n^\alpha}$, which proves the claim. 
\end{proof}

\ifthenelse{1=1}{}{
{\color{red}
{\em Bounding the posterior numerator.} Do the shells thing as usual and then apply the bounds in the assumptions about the loss/risk function:
\begin{align*}
N_n^{(\eta)}(A_n) & = \sum_t \int_{t M_n \eps_n < d(\theta, \theta^\star) < (t+1) M_n \eps_n} e^{-\eta n \{R_n(\theta) - R_n(\theta^\star)\}} \, \Pi(d\theta) \\
& \leq \sum_t e^{-\eta C t^\alpha M_n^\alpha n \eps_n^\alpha} \int_{d(\theta, \theta^\star) < (t+1) M_n \eps_n} e^{\eta n^{1/2} |\GG_n(\ell_\theta - \ell_{\theta^\star})|} \, \Pi(d\theta) \\
& \leq \sum_t e^{-\eta C t^\alpha M_n^\alpha n \eps_n^\alpha + \eta C' n^{1/2} (t+1)^\beta M_n^\beta \eps_n^\beta} \Pi\{d(\theta, \theta^\star) < (t+1) M_n \eps_n\}.
\end{align*}
If the dimension of the $\theta$-space is $q$, then the prior term can be bounded as 
\[ \Pi\{d(\theta, \theta^\star) < (t+1) M_n \eps_n\} \lesssim \{(t+1) M_n \eps_n\}^q. \]
The exponent in the above display is messy, but I'll try to simplify:
\[ -\eta C n M_n^\alpha \eps_n^\alpha t^\alpha \bigl\{ 1 - C'' (t M_n \eps_n)^{\beta-\alpha} n^{-1/2} \bigr\}. \]
Since $\alpha > \beta$ and $n \to \infty$, the term in brackets will eventually be bigger than, say, $\frac12$, so the right-hand side of the big display above can be upper bounded by 
\[ (M_n \eps_n)^q \sum_t (t+1)^q e^{-\eta (C/2) n M_n^\alpha \eps_n^\alpha \, t^\alpha}. \]
The sum is clearly finite for all $n$ and is $o(1)$ as $n \to \infty$.  However, with a little care, I think we can say that it's actually ``exponentially small,'' so that 
\[ N_n^{(\eta)}(A_n) \leq (M_n \eps_n)^q e^{-K \eta n M_n^\alpha \eps_n^\alpha}, \]
and this is with $P$-probability converging to 1 as $n \to \infty$. 
}}

The Gibbs posterior concentration rate result follows directly from Lemmas~\ref{lem:conc_den}--\ref{lem:conc_num}.  Indeed, with $P^n$-probability converging to 1, we have 
\[ \Pi_n^{(\eta)}(A_n) = \frac{N_n^{(\eta)}(A_n)}{D_n^{(\eta)}} \lesssim \frac{(M_n \eps_n)^q e^{-\eta c M_n^\alpha n \eps_n^\alpha}}{\eps_n^q e^{-2\eta b_n n \eps_n^\alpha}} = M_n^q e^{-\eta(c M_n^\alpha - 2b_n) n \eps_n^\alpha}. \]
We are free to choose $b_n$ as small as we like and, if we take $b_n \ll M_n^\alpha$, then the upper bound vanishes, proving Theorem~\ref{thm:conc}.


\bibliographystyle{apalike}
\bibliography{bib}

\end{document}